\documentclass[11pt]{article}
\pdfoutput=1
\usepackage[english]{babel}
\usepackage[utf8]{inputenc}
\usepackage{lastpage}
\usepackage{graphicx}
\usepackage{epstopdf}
\usepackage{pgf,tikz}
\usepackage{mathrsfs}
\usepackage{subfigure}
\usepackage{algorithm}
\usepackage{mathtools}
\usepackage{amsthm}
\usepackage{amsmath}
\usepackage{amssymb}
\usepackage{amsfonts}
\usepackage{calrsfs}
\usepackage{caption}
\usepackage{hyperref}
\usepackage[autostyle=true]{csquotes}
\usepackage[noend]{algpseudocode}
\usepackage{algorithmicx}
\usepackage{csquotes}
\usepackage[colorinlistoftodos]{todonotes}
\usepackage{soul}

\usetikzlibrary{arrows}
\title{\mbox{Probabilistic Analysis} \mbox{of Edge Elimination} \mbox{for Euclidean TSP}}
\author{Xianghui Zhong\\
\small Research Institute for Discrete Mathematics\\
\small University of Bonn\\
\small Lenn\'estr.~2, 53113 Bonn, Germany\\[5mm]
            }
\date{\today} 

\newtheorem{thm}{Theorem}[section]
\theoremstyle{plain}
\newtheorem{lemma}[thm]{Lemma}
\newtheorem*{lemmanonumber}{Lemma}
\newtheorem{theorem}[thm]{Theorem}
\newtheorem{corollary}[thm]{Corollary}

\theoremstyle{remark}
\newtheorem{rem}[thm]{Remark}
\theoremstyle{definition}
\newtheorem{definition}[thm]{Definition}
\newcommand*{\R}{\ensuremath{\mathbb{R}}}

\newcommand*{\N}{\ensuremath{\mathbb{N}}}

\newcommand*{\E}{\ensuremath{\mathbb{E}}}
\DeclareMathOperator{\dist}{dist}
\DeclareMathOperator{\Cert}{Cert}

\begin{document}

\maketitle
\begin{abstract}
One way to speed up the calculation of optimal TSP tours in practice is eliminating edges that are certainly not in the optimal tour as a preprocessing step. In order to do so several edge elimination approaches have been proposed in the past. In this work we investigate two of them in the scenario where the input consists of $n$ independently distributed random points in the 2-dimensional unit square with  density function bounded from above and below by arbitrary positive constants. We show that after the edge elimination procedure of Hougardy and Schroeder the expected number of remaining edges is $\Theta(n)$, while after that the non-recursive part of Jonker and Volgenant the expected number of remaining edges is $\Theta(n^2)$.
\end{abstract}

{\small\textbf{keywords:} traveling salesman problem; Euclidean TSP; probabilistic analysis; edge elimination; preprocessing}

\section{Introduction}
The traveling salesman problem (TSP) is probably the best-known problem in discrete optimization. An instance consists of the pairwise distances between $n$ vertices and the task is to find a shortest Hamilton cycle, i.e.\ a tour visiting every vertex exactly once. The problem is known to be NP-hard \cite{TSPNPHARD}. A special case of the \textsc{TSP} is the \textsc{Euclidean TSP}. Here the vertices are points in the Euclidean plane and distances are given by the $l_2$ norm. This \textsc{TSP} variant is still NP-hard \cite{euclideannpp, euclideannpg}. 

Since the problem is NP-hard, a polynomial-time algorithm is not expected to exist. In order to speed up the calculation of the optimal tour in practice, Jonker and Volgenant developed criteria to find edges that are not contained in any optimal solution, so-called useless edges \cite{Jonker}. The geometrical arguments are based on the following idea: If every tour containing an edge $e$ can be made shorter by replacing at most two edges with two other edges, then $e$ is useless. Thirty years later, Hougardy and Schroeder detect a class of useless edges by showing that every tour containing a specific edge can be shortened by replacing at most three edges \cite{Hougardy}. They are able to give conditions to find useless edges that can be checked efficiently in theory and practice. The algorithm has been tested on instances of TSPLIB, a library for TSP instances \cite{TSPLIB}. In experiments, $30n$ edges remained on average after the execution of the algorithm and the total computation time was improved significantly \cite{Hougardy}.

Although neither elimination procedure is restricted to the 2-dimensional Euclidean case, this paper only focuses on this special case. Both \cite{Hougardy} and \cite{Jonker} contain in addition to the geometric criteria a recursive algorithm for edge elimination. Initially applying the geometric elimination criterion reduces the number of possible edges. Subsequent iterations, combined with other heuristics, will further eliminate additional edges. These parts will not be analyzed in this paper due to complexity.

There have been several probabilistic analyses of algorithms for the TSP. For example the analysis of the 2-Opt Heuristics for TSP by Englert, Röglin and Vöcking \cite{Englert2014}. Starting from an initial tour the 2-Opt Heuristic replaces two edges with two other edges to shorten the tour until it is not possible anymore. In their model each vertex of the TSP instance is a random variable distributed in the $d$ dimensional unit cube by a given probability density function $f_i:[0,1]^d \to [0,\phi]$ bounded from above by a constant $1\leq \phi< \infty$ and they show that for the Euclidean and Manhattan metric the number of expected steps until a local optimum is reached is polynomial in the number of vertices and $\phi$. Moreover, the approximation ratio is bounded by $ O(\sqrt[d]{\phi})$ for all $p$-norms. For $\phi=1$ the input corresponds to uniformly distributed random instances, while for $\phi=\infty$ it corresponds to the worst-case instances. For a suitable choice of the distribution function, worst-case instances perturbed by random noise can also be described by the model. 

In our model we consider $n$ points in the 2-dimensional unit square distributed independently random with densities bounded from above and below by arbitrary positive constants $\phi$ and $\psi$. One obvious density function that satisfies this condition is the uniform distribution on the unit square, with another natural example being the truncated normal distribution on the unit square.

\noindent\textbf{New results.}
In this paper we evaluate the edge elimination criteria of \cite{Hougardy} and \cite{Jonker} on random instances. We show that the expected number of edges that remain after the edge elimination procedure of \cite{Hougardy} is $\Theta(n)$ while the expected number after the procedure in \cite{Jonker} is $\Theta(n^2)$. 
Note that instances with $\Theta(n)$ edges are still NP-hard, as the reduction shown in \cite{euclideannpp} creates instances with $\Theta(n)$ non-useless edges. Nonetheless, this shows that there is a practical preprocessing algorithm that eliminates all but $\Theta(n)$ edges on random instances. 

\noindent\textbf{Structure of the paper.}
After we first describe our model and notation for this paper, we give an outline of the key ideas of the edge elimination procedures from \cite{Hougardy} and \cite{Jonker} and our results.
Then, in the analysis for \cite{Hougardy}, we will develop a modified edge elimination criterion to estimate the probability that an edge can be deleted by the original criterion. For the new criterion we introduce a new construction: a prescribed set of distinct, disk-shaped regions, called test regions, that is a function of the edges. An edge is useless if one of its test regions satisfy a certain condition. We get an estimate on the probability that none of the regions fulfills the edge elimination condition. Using this estimate, it is possible to show that the expected number of remaining edges after the edge elimination procedure is asymptotically linear. In the second part, we investigate the same questions for the procedure described in \cite{Jonker}. For every edge $pq$ and a different vertex $r$, a hyperbola-shaped region is constructed. The edge elimination criterion presented there detects a useless edge if there are no vertices other than $p,q,r$ in the region. We bound the area of this region and show that it contains another vertex with constant probability which then implies that the expected number of the remaining edges is asymptotically quadratic.

\subsection{Model and Notations}
For a Lebesgue measurable set $C\subseteq [0,1]^2$ let the Lebesgue measure of $C$ be denoted by $\lambda(C)$. Moreover, for two points $x,y\in \R^2$ let $\dist(x,y)$ be the Euclidean distance between $x$ and $y$. Similarly, if $x\in \R^2$ and $g$ is a line in the Euclidean plane, let $\dist(x,g)$ be the Euclidean distance of $x$ to the line $g$. For an event $E$ let $E^c$ denote the complementary event of $E$.

\begin{definition}
An instance of the 2-dimensional \textsc{Euclidean TSP} consists of a set of vertices $\{v_1,\dots,v_{n}\}\subset \R^2$. Consider the weighted complete graph $K_n$ with $V(K_n)=\{v_1,\dots,v_{n}\}$, where the cost of an edge $c_{pq}$ is given by $\dist(p,q)$. The 2-dimensional \textsc{Euclidean TSP} asks to find a Hamiltonian cycle, i.e., a cycle that visits every vertex exactly once, of minimal length.
\end{definition}

Due to linearity of the expected value, it is enough to consider a generic edge $pq$. 

\begin{definition}
Given $\psi,\phi\in \R$ with $0<\psi,\phi$, a distribution $d$ on $[0,1]^2$ is said to be \emph{$\psi$-$\phi$-bounded} if $\psi \leq d(x)\leq \phi$ for all $x\in [0,1]^2$.
\end{definition}

\begin{definition}
Given positive constants $\psi$ and $\phi$ consider independent random variables $(U_i)_{i \in \N}$ with $\psi$-$\phi$-bounded density functions $(d_i)_{i\in \N}$ on $[0,1]^2$. Define the random variable 
\begin{align*}
U^n\coloneqq  (U_1,U_2,\dots, U_n).
\end{align*}
The random variable $U^n$ takes values in $([0,1]^2)^n$ and defines a \textsc{Euclidean TSP} instance with $n$ independently distributed vertices in $[0,1]^2$.
\end{definition}

\begin{definition} \label{coredefi}
Let the random variable $X_n^e$ take the value 1 if a specific edge $e$ in the instance represented by $U^n$ remains after the edge elimination process by Hougardy and Schroeder, and 0 otherwise. Let $X_n\coloneqq \sum_{e\in E(K_n)}X_n^e$ be the total number of edges that cannot be deleted on the instance represented by $U^n$. Define $Y_n^e$ and $Y_n$ analogously for the edge elimination process of Jonker and Volgenant.
\end{definition}

Using this notation our main question becomes to calculate $\E\left[X_n\right]$ and $\E\left[Y_n\right]$. 

\section{Outline of the Paper}
In this section we outline the key ideas of the previous edge elimination procedures and the probabilistic analysis of Hougardy and Schroeder.

\subsection{Key Ideas of the Edge Elimination Procedures}
We briefly describe the non-recursive part of the edge elimination procedures of Hougardy and Schroeder \cite{Hougardy} and Jonker and Volgenant \cite{Jonker}. For both edge elimination precedures we fix an optimal tour and assume that an edge $pq$ is in the optimal tour.

\subsubsection{Jonker and Volgenant}
In the paper of Jonker and Volgenant different criteria to detect useless edges were presented. Note that we only analyze the purely geometric part, so we only briefly summarize that part of the edge elimination procedure, which is based on the 2-Opt heuristic. Assume that the edge $pq$ is in the optimal tour $T$ and let $t$ be a neighbor of $r$ in $T$.

The sets
\begin{align*}
I_{p}^{qr}:=\{x\in [0,1]^2 \mid \dist(p,q)+\dist(r,x) \leq \dist(p,r) + \dist(x,q)\} \\
I_{q}^{pr}:=\{x\in [0,1]^2 \mid \dist(p,q)+\dist(r,x) \leq \dist(p,x) + \dist(r,q)\}
\end{align*}
define branches of hyperbolas with $q$, $r$ and $p$, $r$ as foci through $p$ and $q$, respectively (Figure \ref{neighbor of r}). By definition, if $t$ does not lie in $I_{p}^{qr}\cup I_{q}^{pr}$, we can replace $pq$ and $rt$ by either $pr$ and $qt$ or $pt$ and $rq$ to decrease the total cost. It is well known that one of the replacements results in a new tour, contradicting the optimality of $T$ (Figure \ref{2 opt move}).

\begin{figure}[ht]
\centering
 \scalebox{.7}{\definecolor{zzttqq}{rgb}{0.6,0.2,0.}
\begin{tikzpicture}[line cap=round,line join=round,>=triangle 45,x=1.0cm,y=1.0cm]
\clip(-5.423674715724752,-6.905123510391891) rectangle (10.299658249545129,1.6198371014833837);
%\fill[color=zzttqq,fill=zzttqq,fill opacity=0.1] (3.60309331586412,-2.30827564418066) -- (-6.172268152447803,24.13825994955595) -- (-26.633575565272643,-18.132778425386586) -- cycle;
%\fill[color=zzttqq,fill=zzttqq,fill opacity=0.1] (3.60309331586412,-2.30827564418066) -- (18.82955262947165,22.387185518565794) -- (42.57981226987465,-27.86723295031607) -- cycle;
\draw (-4.,-4.)-- (8.,-4.);
%\draw(3.60309331586412,-2.30827564418066) circle (2.0288184634873354cm);
%\draw (-1.7048781222791787,-2.751379946952663) node[anchor=north west] {$R_p(\delta_r)$};
%\draw (6.058925292107195,-2.59914850745489) node[anchor=north west] {$R_q(\delta_r)$};
%\draw [color=zzttqq] (3.60309331586412,-2.30827564418066)-- (-6.172268152447803,24.13825994955595);
%\draw [color=zzttqq] (-6.172268152447803,24.13825994955595)-- (-26.633575565272643,-18.132778425386586);
%\draw [color=zzttqq] (-26.633575565272643,-18.132778425386586)-- (3.60309331586412,-2.30827564418066);
%\draw [color=zzttqq] (3.60309331586412,-2.30827564418066)-- (18.82955262947165,22.387185518565794);
%\draw [color=zzttqq] (18.82955262947165,22.387185518565794)-- (42.57981226987465,-27.86723295031607);
%\draw [color=zzttqq] (42.57981226987465,-27.86723295031607)-- (3.60309331586412,-2.30827564418066);
\draw [color=zzttqq,fill=zzttqq,fill opacity=0.1, samples=50,domain=-0.99:0.99,rotate around={12.544232507963537:(-0.19845334206794057,-3.15413782209033)},xshift=-0.19845334206794057cm,yshift=-3.15413782209033cm] plot ({3.6444363898914167*(1+(\x)^2)/(1-(\x)^2)},{1.373070725231489*2*(\x)/(1-(\x)^2)});
%\draw [samples=50,domain=-0.99:0.99,rotate around={12.544232507963537:(-0.19845334206794057,-3.15413782209033)},xshift=-0.19845334206794057cm,yshift=-3.15413782209033cm] plot ({3.6444363898914167*(-1-(\x)^2)/(1-(\x)^2)},{1.373070725231489*(-2)*(\x)/(1-(\x)^2)});
%\draw [samples=50,domain=-0.99:0.99,rotate around={-21.044394373873807:(5.801546657932056,-3.154137822090328)},xshift=5.801546657932056cm,yshift=-3.154137822090328cm] plot ({2.1054859332064875*(1+(\x)^2)/(1-(\x)^2)},{1.0562238902511936*2*(\x)/(1-(\x)^2)});
\draw [color=zzttqq,fill=zzttqq,fill opacity=0.1, samples=50,domain=-0.99:0.99,rotate around={-21.044394373873807:(5.801546657932056,-3.154137822090328)},xshift=5.801546657932056cm,yshift=-3.154137822090328cm] plot ({2.1054859332064875*(-1-(\x)^2)/(1-(\x)^2)},{1.0562238902511936*(-2)*(\x)/(1-(\x)^2)});
\draw (-1.7048781222791787,-2.751379946952663) node[anchor=north west] {$I_p^{qr}$};
\draw (6.058925292107195,-2.59914850745489) node[anchor=north west] {$I_q^{pr}$};
\begin{scriptsize}
\draw [fill=black] (-4.,-4.) circle (1.5pt);
\draw[color=black] (-3.8578656237476685,-3.6756422581891406) node {$p$};
\draw [fill=black] (8.,-4.) circle (1.5pt);
\draw[color=black] (8.14667074807664,-3.6756422581891406) node {$q$};
\draw [fill=black] (3.60309331586412,-2.30827564418066) circle (1.5pt);
\draw[color=black] (3.753706351140932,-2.00109642371364) node {$r$};
\draw [fill=black] (-6.172268152447803,24.13825994955595) circle (2.5pt);
\draw[color=black] (-5.358432670225707,1.8046895637306792) node {$E$};
\draw [fill=black] (-26.633575565272643,-18.132778425386586) circle (2.5pt);
\draw[color=black] (-5.358432670225707,1.8046895637306792) node {$G$};
\draw [fill=black] (18.82955262947165,22.387185518565794) circle (2.5pt);
\draw[color=black] (-5.358432670225707,1.8046895637306792) node {$H$};
\draw [fill=black] (42.57981226987465,-27.86723295031607) circle (2.5pt);
\draw[color=black] (-5.358432670225707,1.8046895637306792) node {$I$};
\end{scriptsize}
\end{tikzpicture}}
 \caption{For every edge $pq$ and vertex $r\not \in\{p,q\}$ we construct the sets $I_{p}^{qr}$ and $I_{q}^{pr}$. If $pq$ is in the optimal tour, the neighbors of $r$ in the optimal tour have to lie in $I_{p}^{qr} \cup I_{q}^{pr}$.}
 \label{neighbor of r}
\end{figure}

\begin{figure}[ht]
  \centering
   \definecolor{qqffqq}{rgb}{0,1,0}
\definecolor{qqqqff}{rgb}{0,0,1}
\definecolor{ffqqqq}{rgb}{1,0,0}
\begin{tikzpicture}[line cap=round,line join=round,>=triangle 45,x=1.0cm,y=1.0cm]
\draw [color=ffqqqq] (1,1)-- (7,1);
\draw [color=ffqqqq] (3.32,3.48)-- (6.3,4.22);
\draw [color=qqqqff] (1,1)-- (6.3,4.22);
\draw [color=qqqqff] (3.32,3.48)-- (7,1);
\draw [color=qqffqq] (7,1)-- (6.3,4.22);
\draw [color=qqffqq] (3.32,3.48)-- (1,1);
\begin{scriptsize}
\fill [color=black] (1,1) circle (1.5pt);
\draw[color=black] (0.96,1.36) node {$p$};
\fill [color=black] (7,1) circle (1.5pt);
\draw[color=black] (7.16,1.26) node {$q$};
\fill [color=black] (3.32,3.48) circle (1.5pt);
\draw[color=black] (3.44,3.74) node {$r$};
\fill [color=black] (6.3,4.22) circle (1.5pt);
\draw[color=black] (6.5,4.22) node {$t$};
\end{scriptsize}
\end{tikzpicture}
   \caption{If $pq$ and $rt$ are part of a tour, then either replacing the two edges by $pr$ and $qt$ or $pt$ and $rq$ will result in a new tour.}
    \label{2 opt move}
\end{figure}

One neighbor $t$ of $r$ in the optimal tour has to be different from $p$ and $q$ (assuming that we have more than three vertices). Thus, we conclude the edge elimination criterion: 

\begin{theorem}[Non-recursive version of Theorem 1 in \cite{Jonker}] \label{Jonker}
  An edge $pq$ can be eliminated if there is a vertex $r$ such that $I_{p}^{qr}\cup I_{q}^{pr}$ does not contain any other vertex than $p,q,r$.
\end{theorem}

\subsubsection{Hougardy and Schroeder}
An edge elimination procedure for rounded Euclidean instances was described in Hougardy and Schroeder \cite{Hougardy}. Different from the paper we are using exact Euclidean distances. Note that the arguments from \cite{Hougardy} also apply in the exact Euclidean case with slight modifications. We present the results with these modifications.

Given an edge $pq$, a $\delta>0$ and a vertex $x\neq p,q$ the cone $R_p^{qx}(\delta)$ constructed in \cite{Hougardy} plays a central role in the edge elimination procedure. It can be interpreted in the following way if the boundary of $I_p^{qx}$ defined in the previous subsection intersects the circle with radius $\delta$ around $x$ twice: We call the two intersection points $x^p$ and $x^p_2$, where $x^p$ is the point with larger distance to $p$. $R_p^{qx}(\delta)$ is the cone with center $x$ passing through $x^p$ and $x^p_2$ (Figure \ref{hyperbola to cones}). Thus, it contains $I_p^{qx}$ outside of the disk with radius $\delta$ around $r$. The formal definition is given by:

\begin{definition}[Section 3 in \cite{Hougardy}] \label{cool cones}
Given a vertex $r$, for another point $t$ let $t_r(\delta)$ be the intersection of the ray (semi-infinite line) $\overrightarrow{rt}$ with the circle around $r$ with radius $\delta$. For an edge $pq$ we define two cones at $r$: $R_p^{qr}(\delta)\coloneqq \{t\in\R^2\mid \dist(q,t_r(\delta)) \geq \delta+\dist(p,q)-\dist(p,r)\}$ and $R_q^{pr}(\delta)\coloneqq \{t\in\R^2\mid \dist(p,t_r(\delta)) \geq \delta+\dist(p,q)-\dist(q,r)\}$.
\end{definition}

The edge elimination procedure of Hougardy and Schroeder reads as follows:
\begin{figure}[ht]
  \centering
   \scalebox{.7}{\definecolor{zzttqq}{rgb}{0.6,0.2,0.}
\begin{tikzpicture}[line cap=round,line join=round,>=triangle 45,x=1.0cm,y=1.0cm]
\clip(-5.423674715724752,-6.905123510391891) rectangle (10.299658249545129,1.6198371014833837);
\fill[color=zzttqq,fill=zzttqq,fill opacity=0.1] (3.60309331586412,-2.30827564418066) -- (-6.172268152447803,24.13825994955595) -- (-26.633575565272643,-18.132778425386586) -- cycle;
\fill[color=zzttqq,fill=zzttqq,fill opacity=0.1] (3.60309331586412,-2.30827564418066) -- (18.82955262947165,22.387185518565794) -- (42.57981226987465,-27.86723295031607) -- cycle;
\draw (-4.,-4.)-- (8.,-4.);
\draw(3.60309331586412,-2.30827564418066) circle (2.0288184634873354cm);
\draw (-1.7048781222791787,-2.751379946952663) node[anchor=north west] {$R_p^{qx}(\delta)$};
\draw (6.058925292107195,-2.59914850745489) node[anchor=north west] {$R_q^{px}(\delta)$};
\draw [color=zzttqq] (3.60309331586412,-2.30827564418066)-- (-6.172268152447803,24.13825994955595);
\draw [color=zzttqq] (-6.172268152447803,24.13825994955595)-- (-26.633575565272643,-18.132778425386586);
\draw [color=zzttqq] (-26.633575565272643,-18.132778425386586)-- (3.60309331586412,-2.30827564418066);
\draw [color=zzttqq] (3.60309331586412,-2.30827564418066)-- (18.82955262947165,22.387185518565794);
\draw [color=zzttqq] (18.82955262947165,22.387185518565794)-- (42.57981226987465,-27.86723295031607);
\draw [color=zzttqq] (42.57981226987465,-27.86723295031607)-- (3.60309331586412,-2.30827564418066);
\draw [samples=50,domain=-0.99:0.99,rotate around={12.544232507963537:(-0.19845334206794057,-3.15413782209033)},xshift=-0.19845334206794057cm,yshift=-3.15413782209033cm] plot ({3.6444363898914167*(1+(\x)^2)/(1-(\x)^2)},{1.373070725231489*2*(\x)/(1-(\x)^2)});
%\draw [samples=50,domain=-0.99:0.99,rotate around={12.544232507963537:(-0.19845334206794057,-3.15413782209033)},xshift=-0.19845334206794057cm,yshift=-3.15413782209033cm] plot ({3.6444363898914167*(-1-(\x)^2)/(1-(\x)^2)},{1.373070725231489*(-2)*(\x)/(1-(\x)^2)});
%\draw [samples=50,domain=-0.99:0.99,rotate around={-21.044394373873807:(5.801546657932056,-3.154137822090328)},xshift=5.801546657932056cm,yshift=-3.154137822090328cm] plot ({2.1054859332064875*(1+(\x)^2)/(1-(\x)^2)},{1.0562238902511936*2*(\x)/(1-(\x)^2)});
\draw [samples=50,domain=-0.99:0.99,rotate around={-21.044394373873807:(5.801546657932056,-3.154137822090328)},xshift=5.801546657932056cm,yshift=-3.154137822090328cm] plot ({2.1054859332064875*(-1-(\x)^2)/(1-(\x)^2)},{1.0562238902511936*(-2)*(\x)/(1-(\x)^2)});
\draw [style=dashed] (3.60309331586412,-2.30827564418066) -- (3.60309331586412,-2.30827564418066-2.0288184634873354);
\begin{scriptsize}
\draw [fill=black] (-4.,-4.) circle (1.5pt);
\draw[color=black] (-3.8578656237476685,-3.6756422581891406) node {$p$};
\draw [fill=black] (8.,-4.) circle (1.5pt);
\draw[color=black] (8.14667074807664,-3.6756422581891406) node {$q$};
\draw [fill=black] (3.60309331586412,-2.30827564418066) circle (1.5pt);
\draw[color=black] (3.703706351140932,-2.00109642371364) node {$x$};
\draw [fill=black] (-6.172268152447803,24.13825994955595) circle (1.5pt);
\draw [fill=black] (2.8997,-0.40529) circle (1.5pt);
\draw [color=black] (3.1000,-0.20529) node {$x^p$};
\draw [fill=black] (1.80557,-3.24902) circle (1.5pt);
\draw [color=black] (2.10557,-3.37902) node {$x^p_2$};
\draw [color=black] (3.75309331586412,-2.30827564418066-2.0288184634873354/2) node {$\delta$};
%\draw [fill=black] (5.29967,-3.42081) circle (1.5pt);
%\draw [fill=black] (4.66788,-0.58133) circle (1.5pt);
\end{scriptsize}
\end{tikzpicture}}
    \caption{Given an edge $pq$, a $\delta>0$ and a vertex $x\neq p,q$ the cones $R_p^{qx}(\delta)$ and $R_q^{px}(\delta)$ that are constructed from $I_p^{qx}$ and $I_q^{px}$ and the circle with radius $\delta$ around $x$.}
    \label{hyperbola to cones}
  \end{figure}

\begin{theorem}[Main Edge Elimination (for $\delta=\delta_r=\delta_s$), Theorem 3 \& Lemma 12 in \cite{Hougardy}] \label{main elimination}
An edge $pq$ can be eliminated if there exists a pair of vertices $r,s$ disjoint from $p,q$ and a $\delta>0$ that satisfy the following four conditions:
\begin{enumerate}
\item There are no other vertices in each of the disks with radius $\delta$ around $r$ and $s$.
\item The cone $R_p^{qr}(\delta)$ does not contain both neighbors of $r$ in the optimal tour. The analogous statement holds for $R_q^{pr}(\delta)$ and $s$ with the corresponding cones $R_p^{qs}(\delta)$ and $R_q^{ps}(\delta)$.

\item We have $r\not \in R_p^{qs}(\delta) \cup R_q^{ps}(\delta)$ and $s\not \in R_p^{qr}(\delta) \cup R_q^{pr}(\delta)$.

\item The following two inequalities are satisfied (Figure \ref{worst position for the neighbors of rs})
\begin{align*}
\dist(r^p,p)+\dist(r,s)+\dist(s^q,q) &< \dist(r^p,r)+\dist(p,q)+\dist(s, s^q)\\
\dist(s^p,p)+\dist(r,s)+\dist(r^q,q) &< \dist(s^p,s)+\dist(p,q)+\dist(r, r^q).
\end{align*}
\end{enumerate}
\end{theorem}

\begin{minipage}[t]{0.45\textwidth}
 \scalebox{.6}{\definecolor{qqqqff}{rgb}{0,0,1}
\definecolor{ffqqqq}{rgb}{1,0,0}
\begin{tikzpicture}[line cap=round,line join=round,>=triangle 45,x=1.0cm,y=1.0cm]
\clip(-0.45,-5.19) rectangle (12.14,4.25);
\draw [color=ffqqqq] (0,0)-- (10,0);
\draw(4.51,1.41) circle (1.11cm);
\draw(7.04,-1.95) circle (1.11cm);
\draw [domain=-0.44933811931284895:4.511813823716732] plot(\x,{(-5.22--1.03*\x)/-0.41});
\draw [domain=-0.44933811931284895:4.511813823716732] plot(\x,{(--1.98-0.7*\x)/-0.85});
\draw [domain=7.041326548629379:12.13977335467574] plot(\x,{(--12.46-1.88*\x)/0.41});
\draw [domain=7.041326548629379:12.13977335467574] plot(\x,{(-12.46--1.4*\x)/1.33});
\draw [color=qqqqff] (4.51,1.41)-- (7.04,-1.95);
\draw [color=qqqqff] (0,0)-- (4.09,2.43);
\draw [color=qqqqff] (10,0)-- (7.28,-3.03);
\draw [color=ffqqqq] (4.51,1.41)-- (4.09,2.43);
\draw [color=ffqqqq] (7.46,-3.84)-- (7.04,-1.95);
\draw (4.16,3.15) node[anchor=north west] {$r^p$};
\draw (7.42,-2.80) node[anchor=north west] {$s^q$};
\begin{scriptsize}
\fill [color=black] (0,0) circle (1.5pt);
\draw[color=black] (0,0.34) node {$p$};
\fill [color=black] (10,0) circle (1.5pt);
\draw[color=black] (10.14,0.25) node {$q$};
\fill [color=black] (4.51,1.41) circle (1.5pt);
\draw[color=black] (4.63,1.64) node {$r$};
\fill [color=black] (7.04,-1.95) circle (1.5pt);
\draw[color=black] (7.04,-1.58) node {$s$};
\fill [color=black] (4.09,2.43) circle (1.5pt);
\fill [color=black] (7.28,-3.03) circle (1.5pt);
\end{scriptsize}
\end{tikzpicture}}
 \captionof{figure}{Condition (4) states that the blue edges are in total shorter than the red edges.}
  \label{worst position for the neighbors of rs}
\end{minipage}
\hspace{1cm}
\begin{minipage}[t]{0.45\textwidth}
 \scalebox{.6}{\definecolor{qqqqff}{rgb}{0,0,1}
\definecolor{ffqqqq}{rgb}{1,0,0}
\begin{tikzpicture}[line cap=round,line join=round,>=triangle 45,x=1.0cm,y=1.0cm]
\clip(-0.45,-5.47) rectangle (12.12,4.25);
\draw [color=ffqqqq] (0,0)-- (10,0);
\draw(4.51,1.41) circle (1.11cm);
\draw(7.04,-1.95) circle (1.11cm);
\draw [domain=-0.44933811931284895:4.511813823716732] plot(\x,{(-5.22--1.03*\x)/-0.41});
\draw [domain=-0.44933811931284895:4.511813823716732] plot(\x,{(--1.98-0.7*\x)/-0.85});
\draw [domain=7.041326548629379:12.121095147741038] plot(\x,{(--12.46-1.88*\x)/0.41});
\draw [domain=7.041326548629379:12.121095147741038] plot(\x,{(-12.46--1.4*\x)/1.33});
\draw [color=ffqqqq] (2.43,3.53)-- (4.51,1.41);
\draw [color=ffqqqq] (7.04,-1.95)-- (10.59,-4.5);
\draw [color=qqqqff] (0,0)-- (2.43,3.53);
\draw [color=qqqqff] (10.59,-4.5)-- (10,0);
\draw [color=qqqqff] (4.51,1.41)-- (7.04,-1.95);
\begin{scriptsize}
\fill [color=black] (0,0) circle (1.5pt);
\draw[color=black] (0,0.34) node {$p$};
\fill [color=black] (10,0) circle (1.5pt);
\draw[color=black] (10.14,0.25) node {$q$};
\fill [color=black] (4.51,1.41) circle (1.5pt);
\draw[color=black] (4.63,1.64) node {$r$};
\fill [color=black] (7.04,-1.95) circle (1.5pt);
\draw[color=black] (7.04,-1.58) node {$s$};
\fill [color=black] (2.43,3.53) circle (1.5pt);
\draw[color=black] (2.43,3.73) node {$x$};
\fill [color=qqqqff] (10.59,-4.5) circle (1.5pt);
\draw[color=black] (10.89,-4.4) node {$y$};
\end{scriptsize}
\end{tikzpicture}}
  \captionof{figure}{The four conditions ensure that the blue edges are in total shorter than the red edges.}
  \label{two potential points}
\end{minipage}

Assume that the edge $pq$ is in the optimal tour $T$ and we have a pair of vertices $r$ and $s$ satisfying the four conditions. We already know from the previous section that the neighours of $r$ in $T$ have to lie in $I_{q}^{pr} \cup I_{p}^{qr}$. By the first condition, the disk around $r$ does not contain any other vertices. Since the union of the cones $R_p^{qr}(\delta) \cup R_q^{pr}(\delta)$ contains $I_{q}^{pr} \cup I_{p}^{qr}$ outside of the disk with radius $\delta$ around $r$, the neighbors of $r$ has to lie in $R_p^{qr}(\delta) \cup R_q^{pr}(\delta)$.

By the second condition, exactly one of the two neighbors lie in each of $R_p^{qr}(\delta)$ and $R_q^{pr}(\delta)$. Similar statements hold for $s$. 

Let $x$ be the neighbor of $r$ lying in $R_p^{qr}(\delta)$ and $y$ be the neighbor of $s$ lying in $R_q^{ps}(\delta)$ in the optimal tour. By the third statement $r$ and $s$ are not neighbors in the optimal tour $T$, i.e.\ we have $x\neq s$ and $y \neq r$. We want to replace the edges $xp, rs, qy$ by $xr, pq, sy$. (Figure \ref{two potential points}). It can be shown that under the four conditions the term $\dist(x,p) + \dist(r,s) + \dist(q,y) - \dist(x,r) - \dist(p,q) - \dist(s,y)$ is maximal if $x=r^p$ and $y=s^q$. But in this case condition (4) ensures that this term is still negative, i.e.\ after the replacement the total cost decreases. We can use the same argument to show that the symmetrical replacement using the other neighbors of $r$ and $s$ lying in $R_q^{pr}(\delta)$ and $R_p^{qs}(\delta)$ also reduces the total cost. Moreover, one of these two replacements results into a tour, which is a contraction to the optimality of the tour $T$. Thus, the edge $pq$ cannot be in the optimal tour.

\subsection{Probabilistic Analysis of Hougardy and Schroeder}
For the probabilistic analysis of Hougardy and Schroeder, we first develop a new criterion for detecting useless edges. The new criterion detects fewer useless edges than the original, but it makes the analysis easier. The crucial property is that all edges that will be deleted by the new criterion will be deleted by the original criterion as well. In the first part of the analysis our goal is to describe the new criteron and to prove this property.

\subsubsection{A Modified Criterion}
From now on, we fix $\delta \coloneqq \frac{1}{\sqrt{n}}$ and abbreviate $R_p^{qr}(\delta)$ and $R_q^{ps}(\delta)$ by $R_p^{qr}$ and $R_q^{ps}$. A \emph{test region} $T$ for an edge $pq$ is a disk $D$ with radius $2\delta$ aligned similar to the Figure~\ref{outline figure test region}. A test region contains two smaller parts that are marked pink in the figure. That one with smaller distance to $pq$ is called the \emph{lower test subregion} and the other the \emph{upper test subregion}. 

\begin{figure}[ht]
\centering
 \definecolor{zzttqq}{rgb}{0.6,0.2,0}
\begin{tikzpicture}[line cap=round,line join=round,>=triangle 45,x=1.0cm,y=1.0cm]
\draw (4.,-4.)-- (18.,-4.);
\draw [fill=zzttqq, fill opacity=0.1,shift={(12.163913483825889,-2.9029201124954596)}] plot[domain=0.5235987755982929:2.6179938779915,variable=\t]({1.*1.0970798875045185*cos(\t r)+0.*1.0970798875045185*sin(\t r)},{0.*1.0970798875045185*cos(\t r)+1.*1.0970798875045185*sin(\t r)});
\draw (11.213814431265998,-2.354380168743206)-- (13.11401253638578,-2.354380168743206);
\draw [fill=zzttqq, fill opacity=0.1, shift={(12.163913483825889,-2.9029201124954596)}] plot[domain=3.665191429188086:5.759586531581293,variable=\t]({1.*1.0970798875045187*cos(\t r)+0.*1.0970798875045187*sin(\t r)},{0.*1.0970798875045187*cos(\t r)+1.*1.0970798875045187*sin(\t r)});
\draw (11.213814431265998,-3.4514600562477136)-- (13.11401253638578,-3.4514600562477136);
\draw[dashed] (11.613814431265998,-2.354380168743206) -- (11.613814431265998,-3.4514600562477136);
\draw[dashed] (12.163913483825889,-2.9029201124954596) -- (14.3580732588 ,-2.9029201124954596);
\draw[dashed] (12.163913483825889,-2.9029201124954596) -- (12.163913483825889, -1.8058402250084041);
\draw(12.163913483825889,-2.9029201124954596) circle (2.194159775009081cm);
\begin{scriptsize}
\draw [fill=black] (4.,-4.) circle (1.5pt);
\draw[color=black] (4.0924016743771965,-3.756902590628255) node {$p$};
\draw [fill=black] (18.,-4.) circle (1.5pt);
\draw[color=black] (18.10046909794259,-3.756902590628255) node {$q$};
%\draw [fill=black] (7.5,-4.) circle (1.5pt);
%\draw [fill=black] (14.5,-4.) circle (1.5pt);
%\draw [fill=black] (12.163913483825889,-2.9029201124954596) circle (1.5pt);
%\draw[color=black] (12.260330860389269,-2.654989715618196) node {$c$};
\draw[color=black] (13.2609933713,-2.7029201124954596) node {$2\delta$};
\draw[color=black] (11.413814431265998,-2.9029201124954598) node {$\delta$};
\draw[color=black] (12.363913483825889,-2.628650140623695725) node {$\delta$};
\end{scriptsize}
\end{tikzpicture}
  \caption{The test region for the edge $pq$ has the form of a disk with radius $2\delta$. The corresponding upper and lower test subregions are marked pink in the figure.}
  \label{outline figure test region}
\end{figure}
The detailed construction will be described later in the main part. The test regions and test subregions have an area of $\Theta(\frac{1}{n})$, hence the probability that each of them does not contain a vertex is asymptotically constant. The common shape of the test subregions is chosen such that: any two vertices $r$ in the upper test subregion and $s$ in the lower test subregion have at least distance $\delta$ but at most $2\delta$ to each other.

We call a test region \emph{certifying} if there is a vertex $r$ in the upper and a vertex $s$ in the lower test subregion and the test region only contains the two vertices $r$ and $s$.

Next, we call edges of length greater than $c\delta=c\frac{1}{\sqrt{n}}$ \emph{long} edges for some constant $c$. Since the probability of an edge not being long is of order $\Theta(\frac{1}{n})$ it is enough to bound the probability of long edges in order to show that only a linear number of edges remain in expectation. We want to show that if a test region for a long edge is certifying, then the edge $pq$ will be eliminated by the edge elimination procedure of Hougardy and Schroeder.
Assume that a certifying test region is given. By definition, there are vertices $r$ and $s$ in the upper and lower test subregion, respectively.  For $r,s$ to satisfy the edge elimination condition for the edge $pq$ recall that four conditions must be satisfied.

For the first condition note that by the certifying property and the shape of the test region each of the disks with radius $\delta$ around $r$ and $s$ does not contain any other vertex. To exclude that $r$ lies in the disk around $s$ and vice versa we exploit the shape of test subregions that ensures the distance of the upper and lower test subregion to be at least $\delta$. 

For the second condition we show that for long edges $pq$ the angles of the cones will be small. Now, if the two neighbors $x,y$ of $r$ in the optimal tour lie in the same cone, then the angle of $\angle xry$ is also small. Together with the shape of the test region and the certifying property, we ensure that replacing the edges $xr, yr, pq$ by $pr, qr, xy$ gives a shorter tour contradicting the optimality (Figure \ref{near neighbors}).

\begin{figure}[ht]
\centering
 \definecolor{qqqqff}{rgb}{0,0,1}
\definecolor{ffqqqq}{rgb}{1,0,0}
\begin{tikzpicture}[line cap=round,line join=round,>=triangle 45,x=1.0cm,y=1.0cm]
\draw [color=ffqqqq] (1.16,-1.98)-- (10,-2);
\draw [color=ffqqqq] (5.48,0.42)-- (3.2,2.78);
\draw [color=ffqqqq] (5.48,0.42)-- (4.16,2.62);
\draw [color=qqqqff] (3.2,2.78)-- (4.16,2.62);
\draw [color=qqqqff] (1.16,-1.98)-- (5.48,0.42);
\draw [color=qqqqff] (5.48,0.42)-- (10,-2);
\begin{scriptsize}
\fill [color=black] (1.16,-1.98) circle (1.5pt);
\draw[color=black] (1.2,-1.63) node {$p$};
\fill [color=black] (10,-2) circle (1.5pt);
\draw[color=black] (10.11,-1.64) node {$q$};
\fill [color=black] (5.48,0.42) circle (1.5pt);
\draw[color=black] (5.59,0.65) node {$r$};
\fill [color=black] (3.2,2.78) circle (1.5pt);
\draw[color=black] (3.31,3.01) node {$x$};
\fill [color=black] (4.16,2.62) circle (1.5pt);
\draw[color=black] (4.27,2.85) node {$y$};
\end{scriptsize}
\end{tikzpicture}
  \caption{Let $pq$, $xr$ and $yr$ be part of the optimal tour. If the angle $\angle xry$ is small, then the shape of the test region and the certifying property ensure that the red edges are in total shorter than the blue edges.}
  \label{near neighbors}
\end{figure}

For the third condition a straightforward calculation shows that $p\in R_p^{qs}$. We show that for long edges $pq$ the angle $\angle rsp$ is larger than the angles of the cone $R_p^{qs}$. This implies that $r$ does not lie in $R_p^{qs}$. Similarly, we show that $r$ does not lie in $R_q^{ps}$.

The last condition can be proven using geometrical properties of the shape of the test region and test subregion.

\begin{figure}[ht]
\centering
 \begin{tikzpicture}[line cap=round,line join=round,>=triangle 45,x=1.0cm,y=1.0cm]
\draw (0.,-3.)-- (14.,-3.);
\draw [shift={(4.6797346193001115,-2.468999167977505)}] plot[domain=0.523598775598309:2.617993877991484,variable=\t]({1.*0.531000832022491*cos(\t r)+0.*0.531000832022491*sin(\t r)},{0.*0.531000832022491*cos(\t r)+1.*0.531000832022491*sin(\t r)});
\draw (4.219874409337963,-2.2034987519662548)-- (5.1395948292622595,-2.2034987519662548);
\draw [shift={(4.6797346193001115,-2.468999167977505)}] plot[domain=3.665191429188102:5.759586531581277,variable=\t]({1.*0.531000832022491*cos(\t r)+0.*0.531000832022491*sin(\t r)},{0.*0.531000832022491*cos(\t r)+1.*0.531000832022491*sin(\t r)});
\draw (4.219874409337963,-2.734499583988755)-- (5.1395948292622595,-2.734499583988755);
\draw(4.6797346193001115,-2.468999167977505) circle (1.062001664044982cm);
\draw [shift={(9.345709248860153,-2.4523349728719332)}] plot[domain=0.523598775598306:2.6179938779914886,variable=\t]({1.*0.5310008320224903*cos(\t r)+0.*0.5310008320224903*sin(\t r)},{0.*0.5310008320224903*cos(\t r)+1.*0.5310008320224903*sin(\t r)});
\draw (8.885849038898003,-2.186834556860685)-- (9.805569458822301,-2.186834556860685);
\draw [shift={(9.345709248860153,-2.4523349728719332)}] plot[domain=3.6651914291880976:5.75958653158128,variable=\t]({1.*0.5310008320224917*cos(\t r)+0.*0.5310008320224917*sin(\t r)},{0.*0.5310008320224917*cos(\t r)+1.*0.5310008320224917*sin(\t r)});
\draw (8.885849038898003,-2.7178353888831817)-- (9.805569458822301,-2.7178353888831817);
\draw(9.345709248860153,-2.4523349728719332) circle (1.0620016640449834cm);
\draw [shift={(7.0127219340801314,-2.4523349728719337)}] plot[domain=0.5235987755982808:2.6179938779915126,variable=\t]({1.*0.5310008320224947*cos(\t r)+0.*0.5310008320224947*sin(\t r)},{0.*0.5310008320224947*cos(\t r)+1.*0.5310008320224947*sin(\t r)});
\draw (6.552861724117973,-2.1868345568606946)-- (7.47258214404229,-2.1868345568606946);
\draw [shift={(7.0127219340801314,-2.4523349728719337)}] plot[domain=3.6651914291880723:5.759586531581307,variable=\t]({1.*0.5310008320224944*cos(\t r)+0.*0.5310008320224944*sin(\t r)},{0.*0.5310008320224944*cos(\t r)+1.*0.5310008320224944*sin(\t r)});
\draw (6.552861724117973,-2.717835388883172)-- (7.47258214404229,-2.717835388883172);
\draw(7.0127219340801314,-2.4523349728719337) circle (1.0620016640449896cm);
%\draw (3.5,-3.)-- (3.5,-1.406997503932524);
%\draw (3.5,-1.406997503932524)-- (10.5,-1.406997503932524);
%\draw (10.5,-1.406997503932524)-- (10.5,-3.);
\begin{scriptsize}
\draw [fill=black] (0.,-3.) circle (1.5pt);
\draw[color=black] (0.11374516037349916,-2.7022978994555076) node {$p$};
\draw [fill=black] (14.,-3.) circle (1.5pt);
\draw[color=black] (14.111669049053624,-2.7022978994555076) node {$q$};
%\draw [fill=uuuuuu] (3.5,-3.) circle (1.5pt);
%\draw [fill=uuuuuu] (10.5,-3.) circle (1.5pt);
%\draw [fill=black] (4.6797346193001115,-2.468999167977505) circle (1.5pt);
%\draw[color=black] (4.963025936094828,-2.45685607132933) node {$c_1$};
%\draw [fill=uuuuuu] (5.833333333333334,-3.) circle (1.5pt);
%\draw [fill=uuuuuu] (8.166666666666668,-3.) circle (1.5pt);
%\draw [fill=black] (9.345709248860153,-2.4523349728719332) circle (1.5pt);
%\draw[color=black] (9.62900056565487,-2.459021412027361) node {$c_3$};
%\draw [fill=black] (7.0127219340801314,-2.4523349728719337) circle %(1.5pt);
%\draw[color=black] (7.296013250874849,-2.459021412027361) node {$c_2$};
\end{scriptsize}
\end{tikzpicture}
  \caption{Multiple test regions are placed on one side of the edge $pq$}
  \label{outline figure test area}
\end{figure}

Now, we want to define for each edge a certain number of test regions. The practical algorithm in \cite{Hougardy} checks the edge elimination condition for 10 pairs of heuristically chosen $r$ and $s$. We consider the general setting of checking a variable number of pairs. We will show later that checking at most $\omega(\log(n)) \cap o(n)$ pairs is sufficient to conclude that the expected value of remaining edges is linear.

For every long edge $pq$ we place a set $\mathcal{T}=\{T_1, \dots, T_{\lvert \mathcal{T} \rvert} \}$ of test regions on one side of the edge (Figure~\ref{outline figure test area}). It can be shown that we can always choose a side of $pq$ such that each of the test subregions of the placed test regions lies completely inside of the unit square.

Altogether, we conclude our modified edge elimination criterion: A long edge can be eliminated if one of test regions $T_i$ we just placed is certifying. 

\subsubsection{Probabilistic Analysis}
After obtaining the new criterion the next step is to bound the probability that it cannot eliminate an edge, i.e.\ no constructed test region is certifying.
There are two disjoint cases that the $i$th test region is not certifying: 
\begin{itemize}
\item The event $A_i$: the upper or lower test subregion of $T_i$ is empty.
\item The event $B_i$: the upper and lower test subregion of $T_i$ each contains a vertex, but the test region contains at least three vertices.
\end{itemize}
We start by bounding the probability that a given subset of the test regions is not certifying because of $A$ and the remaining test regions are not certifying because of $B$. For that we first check the $A$ events one after another, each conditional on ealier $A$ events. Then, we check the $B$ events one after another, each conditional on all $A$ events and earlier $B$ events. That means we can rewrite this probability by the product of conditional probabilities of the form $\Pr[A_i \mid \land_{j\in J} A_j]$ with $\Pr[B_i \mid \left(\land_{j\in J} A_j \right) \land \left(\land_{k\in K} B_k \right)]$ for some appropriate index sets $J, K$ with $i\not \in J, K$. 
Let us first consider the conditional probability $\Pr[A_i\mid \land_{j\in J} A_j]$. Note that the event $\land_{j\in J} A_j$ means that certain upper or lower test subregions of test regions other than $T_i$ are empty. This only decreases the chance of $A_i$, i.e.\ that the upper or lower test subregion of $T_i$ is empty. Therefore,  we have $\Pr[A_i\mid \land_{j\in J} A_j]\leq \Pr[A_i]$.

Now for the conditional probability $\Pr[B_i \mid (\land_{j\in J} A_j) \land (\land_{k\in K} B_k)]$ note that for some index $j$ both $A_j$ and $B_j$ describe some distribution of vertices in $T_j$. Thus the event $(\land_{j\in J} A_j) \land (\land_{k\in K} B_k)$ influences only the distribution of vertices in $\mathcal{T}\backslash T_i$. Recall that the $B_i$ means that there is a vertex in the upper and lower subregion of $T_i$ and at least three vertices in the test region $T_i$. Hence, $B_i$ is most likely if $\mathcal{T}\backslash{T_i}$ contains no other vertices, since this increases the probability that vertices lie in $T_i$. Therefore, we get $\Pr[B_i \mid (\land_{j\in J} A_j) \land (\land_{k\in K} B_k)]\leq \Pr[B_i \mid \mathcal{T}\backslash T_i\text{ contains no vertices}]$. 

This gives us an estimate for the probability that a given subset of the test regions are not certifying because of $A$ and the remaining test regions are not certifying because of $B$. By summing up all possible subsets we can bound the probability that no test region is certifying by $\Pr[X_n^{pq}=1]\leq (\Pr[A_i]+\Pr[B_i \mid \mathcal{T}\backslash T_i\text{ contains no vertices}])^{\lvert \mathcal{T} \rvert}$. 

Next, we want to bound the probability that a test region is not certifying by a constant $s<1$. Note that the areas of a test region and a test subregion are proportional to $\delta^2=\frac{1}{n}$. Recall that we denote the Lebesgue measure of a set $C$ by $\lambda(C)$. By limiting the number of constructed test regions by $\lvert \mathcal{T} \rvert \in o(n)$, we ensure that $\lim_{n\to\infty} \lambda(\mathcal{T}\backslash T_i) =\lim_{n\to\infty} \lvert \mathcal{T}\backslash T_i \rvert \lambda(T_i)= \lim_{n\to \infty}(\lvert \mathcal{T} \rvert-1)\Theta\left(\frac{1}{n}\right)= 0$. This indicates due to the bounded densities that $\lim_{n\to\infty}\Pr[B_i \mid \mathcal{T} \backslash T_i\text{ contains no vertices}] = \lim_{n\to \infty}\Pr[B_i]$. Moreover, the probability that each of the the test subregions contain exactly one vertex and there are no other vertices in the test region, i.e.\ the test region is certifying, is in $\Theta\left((1-\frac{1}{n})^n\right)=\Theta\left(1\right)$. Thus, we have for the complement event $\limsup_{n\to \infty}\Pr[A_i]+\Pr[B_i]<1$. Altogether, we have 
\begin{align*}
\limsup_{n\to \infty}\Pr[A_i]+\Pr[B_i \mid \mathcal{T}\backslash T_i \text{ contains no vertices}]=
\limsup_{n\to \infty}\Pr[A_i]+\Pr[B_i]<s<1
\end{align*} for some constant $s$. Hence, the probability that none of $\lvert \mathcal{T} \rvert$ test regions are certifying can be bounded by $\Pr[X_n^{pq}=1] \leq s^{\lvert \mathcal{T} \rvert}$.

Now the probability that $pq$ remains is:
\begin{align*}
\Pr[X_n^{pq}=1]=\sum_{l} \Pr[X_n^{pq}=1 \mid \lvert \mathcal{T} \rvert=l]\Pr[\lvert \mathcal{T} \rvert=l].
\end{align*}
We already have $\Pr[X_n^{pq}=1 \mid \lvert \mathcal{T} \rvert=l] \leq s^l$. Recall that all vertices are randomly distributed with a $\psi$-$\phi$-bounded density function for some constants $\psi$ and $\phi$. It can be shown that for an edge $pq$ formed by two random vertices $p$ and $q$ we have $\Pr[\lvert \mathcal{T} \rvert=l] \in O\left(\delta^2\phi\right)=O\left(\frac{1}{n}\right)$. Hence, $n\Pr[X_n^{pq}=1] \in O\left(n\sum_{l}s^l\frac{1}{n}\right)=O\left(\sum_{l}s^l\right)=O\left(1\right)$ as $s<1$. Thus, we conclude:
\begin{align*}
\limsup_{n\to \infty} \frac{\E[X_n]}{n}= \limsup_{n\to \infty} \frac{(n-1)\Pr[X_n^{pq}=1]}{2}<\infty.
\end{align*}
Therefore, $\E[X_n]\in O(n)$. On the other hand, we have $\E[X_n]\in \Omega(n)$, since an optimal \textsc{TSP} tour exists for every instance and consists of $n$ edges which cannot be eliminated. Therefore, we conclude $\E[X_n]\in \Theta(n)$.

\section{Probabilistic Analysis of Jonker and Volgenant}
In this section we analyze the expected number of edges that are not eliminated by the criterion of Jonker and Volgenant. We consider a fixed edge $pq$. Moreover, we assume that $r$ is a vertex other than $p$ or $q$. 

Recall that the edge elimination procedure of Jonker and Volgenant eliminates the edge $pq$ if for another vertex $r$ the set $I_q^{pr}\cup I_p^{qr}$ does not contain any vertex other than $p,q$ and $r$. Our goal is to get a bound on the area of $I_q^{pr}\cup I_p^{qr}$ and to show that it is likely that they contain vertices other than $p,q$ and $r$. This leads to the result that the expected number of remaining edges is quadratic.

Let $0<\alpha<\frac{1}{2}$ and $0<2\alpha<\beta<1$ be two fixed parameters whose values we will determine later. Let $B_{pq}$ be the event where $p$ lies in the square $\{(x,y)\mid \frac{1-\beta}{2} \leq x \leq \frac{1-\beta}{2}+\alpha, \frac{1-\alpha}{2} \leq y \leq \frac{1+\alpha}{2} \}$ and $q$ lies in the square $\{(x,y)\mid \frac{1+\beta}{2}-\alpha \leq x \leq \frac{1+\beta}{2}, \frac{1-\alpha}{2} \leq y \leq \frac{1+\alpha}{2}\}$ (Figure~\ref{pqbox}). By the choice $2\alpha<\beta$, the two squares are disjoint. If $\alpha =0$ and $B_{pq}$, then the line $pq$ is the horizontal line $y=\frac{1}{2}$. Therefore, for fixed $\beta$ we can choose $\alpha>0$ small enough such that if $B_{pq}$ happens, then for all possible positions of $p$ and $q$ the line $pq$ intersects the interior of the left and right border of the unit square. 

We say that $r$ is \emph{centrally located} if $\frac{1-\beta}{2} \leq r_x \leq \frac{1+\beta}{2}$, where $r_x$ is the $x$-coordinate of $r$.

\begin{figure}[h]
\begin{minipage}[t]{0.45\textwidth}
  \center
\begin{tikzpicture}[line cap=round,line join=round,>=triangle 45,x=5cm,y=5cm]
\draw [line width=1pt] (0,1)-- (0,0);
\draw [line width=1pt] (0,0)-- (1,0);
\draw [line width=1pt] (1,0)-- (1,1);
\draw [line width=1pt] (1,1)-- (0,1);
\draw [line width=1pt] (0.2,0.45)-- (0.2,0.55);
\draw [line width=1pt] (0.3,0.45)-- (0.3,0.55);
\draw [line width=1pt] (0.2,0.45)-- (0.3,0.45);
\draw [line width=1pt] (0.2,0.55)-- (0.3,0.55);

\draw [line width=1pt] (0.7,0.45)-- (0.7,0.55);
\draw [line width=1pt] (0.8,0.45)-- (0.8,0.55);
\draw [line width=1pt] (0.7,0.45)-- (0.8,0.45);
\draw [line width=1pt] (0.7,0.55)-- (0.8,0.55);

\draw [line width=1pt, dashed] (0.0,0.45)-- (0.2,0.45);
\draw [line width=1pt, dashed] (0.8,0.45)-- (1.0,0.45);

\draw[color=black, below] (0.1,0.45) node {$\frac{1-\beta}{2}$};
\draw[color=black, below] (0.9,0.45) node {$\frac{1-\beta}{2}$};

\draw [line width=1pt, dashed] (0.3,0.45)-- (0.7,0.45);

\draw[color=black, below] (0.5,0.45) node {$\beta-2\alpha$};

\draw [line width=1pt] (0.24451704383848787,0.48534839111861154)-- (0.7633789319839351,0.50611149081180348);
\begin{scriptsize}
\draw [fill=black] (0.24451704383848787,0.48534839111861154) circle (1.5pt);
\draw[color=black, above] (0.24451704383848787,0.48534839111861154) node {$p$};
\draw [fill=black] (0.7633789319839351,0.50611149081180348) circle (1.5pt);
\draw[color=black, below] (0.7633789319839351,0.50611149081180348) node {$q$};
\draw[color=black, left] (0.2,0.5) node {$\alpha$};
\draw[color=black, above] (0.25,0.55) node {$\alpha$};
\end{scriptsize}
\end{tikzpicture}
\caption{The event $B_{pq}$: $p$ and $q$ each lies in a square with side length $\alpha$. Moreover, $\alpha$ is chosen small enough such that if $B_{pq}$ happens, then the line $pq$ intersects the interior of the left and right border of the unit square for all possible positions of $p,q$.}
\label{pqbox}
\end{minipage}\qquad% 
\begin{minipage}[t]{0.45\textwidth}
  \center
  \definecolor{zzttqq}{rgb}{0.6,0.2,0.}
  \begin{tikzpicture}[line cap=round,line join=round,>=triangle 45,x=5cm,y=5cm]
  \fill[line width=1pt,color=zzttqq,fill=zzttqq,fill opacity=0.10000000149011612] (0,0.45754984676955407) -- (0.25867234810055695,0.47597123456042034) -- (0.41966987886582324,0.627826751553825) -- (0,0.7729120613226795) -- cycle;
  \fill[line width=1pt,color=zzttqq,fill=zzttqq,fill opacity=0.10000000149011612] (0.41966987886582324,0.627826751553825) -- (0.7563947476050682,0.5114166055125129) -- (1,0.5287649879554417) -- (1,1) -- (0.8142487206692763,1) -- cycle;
  \draw [line width=1pt,color=black] (0,1)-- (0,0);
  \draw [line width=1pt,color=black] (0,0)-- (1,0);
  \draw [line width=1pt,color=black] (1,0)-- (1,1);
  \draw [line width=1pt,color=black] (1,1)-- (0,1);
  \draw [line width=1pt,color=black] (0.25867234810055695,0.47597123456042034)-- (0.7563947476050682,0.5114166055125129);
  \draw [line width=1pt,color=black] (0.2,0.55)-- (0.2,0.45);
  \draw [line width=1pt,color=black] (0.2,0.45)-- (0.3,0.45);
  \draw [line width=1pt,color=black] (0.3,0.45)-- (0.3,0.55);
  \draw [line width=1pt,color=black] (0.3,0.55)-- (0.2,0.55);
  \draw [line width=1pt,color=black] (0.7,0.55)-- (0.7,0.45);
  \draw [line width=1pt,color=black] (0.7,0.45)-- (0.8,0.45);
  \draw [line width=1pt,color=black] (0.8,0.45)-- (0.8,0.55);
  \draw [line width=1pt,color=black] (0.8,0.55)-- (0.7,0.55);
  \draw [line width=1pt,color=black] (0.7563947476050682,0.5114166055125129)-- (0,0.7729120613226795);
  \draw [line width=1pt,color=black] (0.25867234810055695,0.47597123456042034)-- (0.8142487206692763,1);
  \draw [line width=1pt,color=black] (0,0.45754984676955407)-- (1,0.5287649879554417);
  \draw [line width=1pt,color=black] (0.15,0.6)--(0.7563947476050682,0.5114166055125129);
  \begin{scriptsize}
  \draw [fill=black] (0.15,0.6) circle (1.5pt);
  \draw[color=black, above] (0.15,0.6) node {$x$};
  \draw [fill=black] (0.25867234810055695,0.47597123456042034) circle (1.5pt);
  \draw[color=black, above] (0.25867234810055695,0.47597123456042034) node {$p$};
  \draw [fill=black] (0.7563947476050682,0.5114166055125129) circle (1.5pt);
  \draw[color=black, below] (0.7563947476050682,0.5114166055125129) node {$q$};
  \draw [fill=black] (0.41966987886582324,0.627826751553825) circle (1.5pt);
  \draw[color=black, above] (0.41966987886582324,0.627826751553825) node {$r$};
  \draw [fill=black] (0,0.7729120613226795) circle (1.5pt);
  \draw[color=black, left] (0,0.7729120613226795) node {$\widetilde{q}$};
  \draw [fill=black] (0.8142487206692763,1) circle (1.5pt);
  \draw[color=black, above] (0.8142487206692763,1) node {$\widetilde{p}$};
  \draw [fill=black] (0,0.45754984676955407) circle (1.5pt);
  \draw[color=black, left] (0,0.45754984676955407) node {$u$};
  \draw [fill=black] (1,0.5287649879554417) circle (1.5pt);
  \draw[color=black, right] (1,0.5287649879554417) node {$v$};
  \draw[color=black, above] (0.08,0.55) node {$P$};
  \draw[color=black] (0.77,0.7) node {$Q$};
\end{scriptsize}
  \end{tikzpicture}
\caption{Let $x$ be a point in the area $P$ bounded by $up, pr, r\widetilde{q}$ and the border of the unit square. The segments $qx$ and $pr$ cross each other. Similarly, if $x$ is a point in the area $Q$ bounded by $qv, rq, r\widetilde{p}$ and the border of the unit square the segments $px$ and $qr$ cross each other.}
\label{pqbox intersect}
\end{minipage}
\end{figure}

\begin{figure}[h]
\begin{minipage}[t]{0.45\textwidth}
    \center
    \definecolor{zzttqq}{rgb}{0.6,0.2,0.}
    \begin{tikzpicture}[line cap=round,line join=round,>=triangle 45,x=5cm,y=5cm]
    \fill[line width=1pt,color=zzttqq,fill=zzttqq,fill opacity=0.10000000149011612] (0,0.45754984676955407) -- (0.25867234810055695,0.47597123456042034) -- (0.41966987886582324,0.627826751553825) -- (0,0.7729120613226795) -- cycle;
    \fill[line width=1pt,color=zzttqq,fill=zzttqq,fill opacity=0.10000000149011612] (0.41966987886582324,0.627826751553825) -- (0.7563947476050682,0.5114166055125129) -- (1,0.5287649879554417) -- (1,1) -- (0.8142487206692763,1) -- cycle;
    \fill[line width=1pt,color=zzttqq,fill=zzttqq,fill opacity=0.10000000149011612] (0.7563947476050682,0.5114166055125129) -- (1,0.5287649879554417) -- (1,1) -- cycle;
    \draw [line width=1pt,color=black] (0,1)-- (0,0);
    \draw [line width=1pt,color=black] (0,0)-- (1,0);
    \draw [line width=1pt,color=black] (1,0)-- (1,1);
    \draw [line width=1pt,color=black] (1,1)-- (0,1);
    \draw [line width=1pt,color=black] (0.25867234810055695,0.47597123456042034)-- (0.7563947476050682,0.5114166055125129);
    \draw [line width=1pt,color=black] (0.2,0.55)-- (0.2,0.45);
    \draw [line width=1pt,color=black] (0.2,0.45)-- (0.3,0.45);
    \draw [line width=1pt,color=black] (0.3,0.45)-- (0.3,0.55);
    \draw [line width=1pt,color=black] (0.3,0.55)-- (0.2,0.55);
    \draw [line width=1pt,color=black] (0.7,0.55)-- (0.7,0.45);
    \draw [line width=1pt,color=black] (0.7,0.45)-- (0.8,0.45);
    \draw [line width=1pt,color=black] (0.8,0.45)-- (0.8,0.55);
    \draw [line width=1pt,color=black] (0.8,0.55)-- (0.7,0.55);
    \draw [line width=1pt,color=black] (0.7563947476050682,0.5114166055125129)-- (0,0.7729120613226795);
    \draw [line width=1pt,color=black] (0.25867234810055695,0.47597123456042034)-- (0.8142487206692763,1);
    \draw [line width=1pt,color=black] (0,0.45754984676955407)-- (1,0.5287649879554417);
    \draw [line width=1pt,color=black] (0.7563947476050682,0.5114166055125129)-- (1,1);
    \begin{scriptsize}
    \draw [fill=black] (0.25867234810055695,0.47597123456042034) circle (1.5pt);
    \draw[color=black, above] (0.25867234810055695,0.47597123456042034) node {$p$};
    \draw [fill=black] (0.7563947476050682,0.5114166055125129) circle (1.5pt);
    \draw[color=black, below] (0.7563947476050682,0.5114166055125129) node {$q$};
    \draw [fill=black] (0.41966987886582324,0.627826751553825) circle (1.5pt);
    \draw[color=black, above] (0.41966987886582324,0.627826751553825) node {$r$};
    \draw [fill=black] (0,0.7729120613226795) circle (1.5pt);
    \draw[color=black, left] (0,0.7729120613226795) node {$\widetilde{q}$};
    \draw [fill=black] (0.8142487206692763,1) circle (1.5pt);
    \draw[color=black, above] (0.8142487206692763,1) node {$\widetilde{p}$};
    \draw [fill=black] (1,1) circle (2pt);
    \draw[color=black, above] (1,1) node {$t$};
    \draw [fill=black] (0,0.45754984676955407) circle (1.5pt);
    \draw[color=black, left] (0,0.45754984676955407) node {$u$};
    \draw [fill=black] (1,0.5287649879554417) circle (1.5pt);
    \draw[color=black, right] (1,0.5287649879554417) node {$v$};
    \end{scriptsize}
    \end{tikzpicture}
  \caption{If $\widetilde{p}$ lies on the top border of the unit square, the triangle $qvt$ lies in $Q$.}
  \label{pqbox intersect top}
\end{minipage}\qquad% 
\begin{minipage}[t]{0.45\textwidth}
    \center
    \definecolor{zzttqq}{rgb}{0.6,0.2,0.}
    \begin{tikzpicture}[line cap=round,line join=round,>=triangle 45,x=5cm,y=5cm]
    \fill[line width=1pt,color=zzttqq,fill=zzttqq,fill opacity=0.10000000149011612] (0,0.45754984676955407) -- (0.25867234810055695,0.47597123456042034) -- (0.4488449060498531,0.5945692862360038) -- (0,0.7159241111268112) -- cycle;
    \fill[line width=1pt,color=zzttqq,fill=zzttqq,fill opacity=0.10000000149011612] (0.7563947476050682,0.5114166055125129) -- (1,0.5287649879554417) -- (1,0.9382882802823922) -- (0.4488449060498531,0.5945692862360038) -- cycle;
    \fill[line width=1pt,color=zzttqq,fill=zzttqq,fill opacity=0.10000000149011612] (0,0.45754984676955407) -- (0.25867234810055695,0.47597123456042034) -- (0.4488449060498531,0.5945692862360038) -- (0,0.562604732881097) -- cycle;
    \fill[line width=1pt,color=zzttqq,fill=zzttqq,fill opacity=0.10000000149011612] (0.7563947476050682,0.5114166055125129) -- (1,0.5287649879554417) -- (1,0.6338198740669847) -- (0.4488449060498531,0.5945692862360038) -- cycle;
    \draw [line width=1pt,color=black] (0,1)-- (0,0);
    \draw [line width=1pt,color=black] (0,0)-- (1,0);
    \draw [line width=1pt,color=black] (1,0)-- (1,1);
    \draw [line width=1pt,color=black] (1,1)-- (0,1);
    \draw [line width=1pt,color=black] (0.25867234810055695,0.47597123456042034)-- (0.7563947476050682,0.5114166055125129);
    \draw [line width=1pt,color=black] (0.2,0.55)-- (0.2,0.45);
    \draw [line width=1pt,color=black] (0.2,0.45)-- (0.3,0.45);
    \draw [line width=1pt,color=black] (0.3,0.45)-- (0.3,0.55);
    \draw [line width=1pt,color=black] (0.3,0.55)-- (0.2,0.55);
    \draw [line width=1pt,color=black] (0.7,0.55)-- (0.7,0.45);
    \draw [line width=1pt,color=black] (0.7,0.45)-- (0.8,0.45);
    \draw [line width=1pt,color=black] (0.8,0.45)-- (0.8,0.55);
    \draw [line width=1pt,color=black] (0.8,0.55)-- (0.7,0.55);
    \draw [line width=1pt,color=black] (0.7563947476050682,0.5114166055125129)-- (0,0.7159241111268112);
    \draw [line width=1pt,color=black] (0,0.45754984676955407)-- (1,0.5287649879554417);
    \draw [line width=1pt,color=black] (0.25867234810055695,0.47597123456042034)-- (0.4488449060498531,0.5945692862360038);
    \draw [line width=1pt,color=black] (0,0.562604732881097)-- (1,0.6338198740669847);
    \draw [line width=1pt,color=black] (0.4488449060498531,0.5945692862360038)-- (1,0.9382882802823922);
    \begin{scriptsize}
    \draw [fill=black] (0.25867234810055695,0.47597123456042034) circle (1.5pt);
    \draw[color=black, above] (0.25867234810055695,0.47597123456042034) node {$p$};
    \draw [fill=black] (0.7563947476050682,0.5114166055125129) circle (1.5pt);
    \draw[color=black, below] (0.7563947476050682,0.5114166055125129) node {$q$};
    \draw [fill=black] (0.4488449060498531,0.5945692862360038) circle (1.5pt);
    \draw[color=black, above] (0.4488449060498531,0.5945692862360038) node {$r$};
    \draw [fill=black] (0,0.7159241111268112) circle (1.5pt);
    \draw[color=black, left] (0,0.7159241111268112) node {$\widetilde{q}$};
    \draw [fill=black] (0,0.45754984676955407) circle (1.5pt);
    \draw[color=black, left] (0,0.45754984676955407) node {$u$};
    \draw [fill=black] (1,0.5287649879554417) circle (1.5pt);
    \draw[color=black, right] (1,0.5287649879554417) node {$v$};
    \draw [fill=black] (1,0.9382882802823922) circle (1.5pt);
    \draw[color=black, right] (1,0.9382882802823922) node {$\widetilde{p}$};
    \draw [fill=black] (0,0.562604732881097) circle (1.5pt);
    \draw[color=black, left] (0,0.562604732881097) node {$a$};
    \draw [fill=black] (1,0.6338198740669847) circle (1.5pt);
    \draw[color=black, right] (1,0.6338198740669847) node {$b$};
    \end{scriptsize}
    \end{tikzpicture}
  \caption{In the case that $\widetilde{q}$ and $\widetilde{p}$ lie on the left and right border of the unit square: The trapezoids $upra$ and $qvbr$ lie in $P$ and $Q$, respectively.}
  \label{pqbox intersect left right}
\end{minipage}
\end{figure}

\begin{lemma} \label{at least area of improvement cone}
  There is a constant $f$ such that if $B_{pq}$, then $\lambda(I^{pr}_q \cup I^{qr}_p) \geq \min\left(\frac{\dist(r,pq)}{2}, f\right)$. Moreover, if $r$ is not centrally located, then we have $\lambda(I^{pr}_q \cup I^{qr}_p ) \geq f$.
\end{lemma}
  
\begin{proof}
  Let $\widetilde{p}$ and $\widetilde{q}$ be the intersection points of the rays $\overrightarrow{pr}$ and $\overrightarrow{qr}$ with the border of the unit square, respectively. Similarly, let $u$ and $v$ be the intersection points of the rays $\overrightarrow{qp}$ and $\overrightarrow{pq}$ with the border of the unit square, respectively. Let $x$ be a point in the area $P$ bounded by $up, pr, r\widetilde{q}$ and the border of the unit square. Then, by construction the segments $qx$ and $pr$ cross each other (Figure~\ref{pqbox intersect}). By applying the triangle inequality twice, we get $\dist(q,x)+\dist(p,r)\geq\dist(x,r)+\dist(p,q)$. Thus, $P$ is a subset of $I^{qr}_p $. Similarly, the area $Q$ bounded by $qv, rq, r\widetilde{p}$ and the border of the unit square is a subset of $I^{pr}_q$. 

  If $\widetilde{p}$ does not lie on the right border of the unit square, assume w.l.o.g.\ that it lies on the top border of the unit square. Let $t=(1,1)$ be the top right border of the unit square. The triangle with the endpoints $t$, $q$ and $v$ lies completely inside of $Q$ (Figure~\ref{pqbox intersect top}). Recall that we chose $\alpha$ small enough such that if $B_{pq}$, then for all possible positions of $p,q$ the line $pq$ intersects the interior of the right border of the unit square. Hence, there are constants $c_1, c_2>0$ such that $\dist(v,t)\geq c_1$ and $\sin(\angle tvq) \geq c_2$. Therefore the area of the triangle is at least the constant $\frac{1}{2}\dist(q, v)\dist(v,t)\sin(\angle tvq) \geq \frac{1}{2}\frac{1-\beta}{2}c_1c_2=:f$. The analogous statement holds if $\widetilde{q}$ does not lie on the left border of the unit square. In particular, if $r$ is not centrally located one of these two cases has to occur, which shows the second statement.

  In the case that $\widetilde{p}$ lies on the right border and $\widetilde{q}$ lies on the left border of the unit square, consider the parallel line to $pq$ through $r$. Let the intersection point of this parallel line with the left and right border unit square be $a$ and $b$, respectively. Then the trapezoids $upra$ and $qvbr$ lie in $P$ and $Q$, respectively (Figure \ref{pqbox intersect left right}). Their combined area is at least
  \begin{align*} 
    &\frac{\left(\dist(a,r)+\dist(u,p)\right)\dist(r,pq)}{2}+\frac{\left(\dist(r,b)+\dist(q,v)\right)\dist(r,pq)}{2}\\
    = &\frac{\left(\dist(a,b)+\dist(u,p)+\dist(q,v)\right)\dist(r,pq)}{2} \geq \frac{\dist(r,pq)}{2}.
  \end{align*}
\end{proof}

Using the upper bounds on the area of $I^{pr}_q$ and $I^{qr}_p$, we next investigate on the asymptotic behavior of the probability that these areas are empty.

\begin{definition}
Let $\Cert(r)$ be the event that the vertex $r$ fulfills the condition of Theorem \ref{Jonker} and we can therefore eliminate $pq$.
\end{definition}

\begin{lemma} \label{prob area prob}
There is a is a constant $c$ such that
\begin{align*}
\lim_{n\to \infty} n\Pr\left[\Cert(r) \mid B_{pq}\right]& \leq c\beta.\\
\end{align*}
\end{lemma}

\begin{proof}
If $\Cert(r)$ happens, then $I^{pr}_q \cup I^{qr}_p$ does not contain the remaining $n-3$ vertices. W.l.o.g.\ assume that $p=U_{n-2},q=U_{n-1}$ and $r=U_n$. By Lemma \ref{at least area of improvement cone} we can bound $\lambda(I^{pr}_q \cup I^{qr}_p)$ either by $f$ or $\frac{\dist(r,pq)}{2}$. We have $\lambda(I^{pr}_q \cup I^{qr}_p)\geq f$ if $r$ is not centrally located with probability at least $1-\beta\phi$. Thus, there is a $0 \leq y\leq \beta\phi$ such that:
\begin{align*}
  \lim_{n\to \infty}n\Pr\left[\Cert(r) \mid B_{pq} \right]&=\lim_{n\to \infty}n \prod_{i=1}^{n-3}\int_{[0,1]^2\backslash I^{pr}_q \cup I^{qr}_p}d_i\mathrm{d}\lambda=\lim_{n\to \infty}n \prod_{i=1}^{n-3}\left( 1-\int_{I^{pr}_q \cup I^{qr}_p}d_i\mathrm{d}\lambda \right)\\
 &\leq \lim_{n\to \infty}n (1-\lambda(I^{pr}_q \cup I^{qr}_p)\psi)^{n-3}\\
 &\leq \lim_{n\to \infty} \left( n(1-f\psi)^{n-3} \left(1-y\right) + n(1-\frac{1}{2}\dist(r, pq)\psi)^{n-3} y \right)\\
  &\leq \lim_{n\to\infty} n (1-\frac{1}{2}\dist(r, pq)\psi)^{n-3} \beta\phi.
\end{align*}
Now, $\dist(r, pq)$ can vary between 0 and $\sqrt{2}$. For fixed $p,q$ and fixed distance $d$ the position of $r$ satisfying $\dist(r,pq)=d$ is a segment in the unit square with length at most $\sqrt{2}$. Therefore, 
\begin{align*}
  \lim_{n\to \infty}n\Pr\left[\Cert(r) \mid B_{pq} \right]&\leq \beta\phi \lim_{n\to\infty} n\int_{0}^{\sqrt{2}}(1-\frac{1}{2}x\psi)^{n-3} \sqrt{2}\phi \mathrm{d}x \\
   &= \beta\phi^2 \sqrt{2} \lim_{n\to\infty} n \left[-\frac{(1-\frac{1}{2}x\psi)^{n-2}}{\frac{1}{2}\psi(n-2)}\right]_0^{\sqrt{2}}\\
   &= \beta\phi^2 \sqrt{2} \lim_{n\to\infty} \left( \frac{n}{\frac{1}{2}\psi(n-2)}-\frac{n(1-\frac{1}{2}\sqrt{2}\psi)^{n-2}}{\frac{1}{2}\psi(n-2)} \right)=\frac{\beta\phi^2 \sqrt{2}}{\frac{1}{2}\psi}.
\end{align*}
\end{proof}

\begin{theorem} \label{Jonker result}
$\E\left[Y_n\right]\in \Theta(n^2)$
\end{theorem}

\begin{proof}
We have:
\begin{align*}
\E\left[Y_n\right]& = \sum_{e\in E(K_n)} \E\left[Y^e_n\right]=\sum_{e\in E(K_n)} \Pr\left[Y^e_n=1\right]=\binom{n}{2} \Pr\left[Y^{pq}_n=1\right]\\
&\geq \binom{n}{2} \left(1-\Pr\left[Y^{pq}_n=0 \mid B_{pq}\right]\right)\Pr\left[B_{pq}\right].
\end{align*}
By choosing some $\beta<\frac{1}{c}$ we have by union bound and Lemma \ref{prob area prob}:
\begin{align*}
\liminf_{n\to \infty} \frac{\E\left[Y_n\right]}{n^2} &\geq \frac{1}{2} \Pr\left[B_{pq}\right] \left(1-\lim_{n\to \infty}  \Pr\left[Y^{pq}_n=0 \mid B_{pq}\right]\right)\\
&= \frac{1}{2} \Pr\left[B_{pq}\right] \left(1-\lim_{n\to \infty}  \Pr\left[\lor_{r\in V(K_n)\backslash \{p,q\}} \Cert(r) \mid B_{pq}\right]\right)\\
&\geq \frac{1}{2} \Pr\left[B_{pq}\right] \left(1-\lim_{n\to \infty} n\Pr\left[\Cert(r) \mid B_{pq}\right]\right)\\
&\geq \frac{1}{2} (\alpha\psi)^2 \left(1-c\beta \right)>0.
\end{align*}
Therefore, we have $\E\left[Y_n\right] \in \Omega(n^2)$. On the other hand, we have $\E\left[Y_n\right]\in O(n^2)$, since the graph $K_n$ has $\binom{n}{2}$ edges. This proves $\E\left[Y_n\right]\in \Theta(n^2)$.
\end{proof}

\begin{rem}
A computer program was written to analyze the probability that an edge is eliminated in practice by the criterion of Jonker and Volgenant for uniformly distributed instances.
The results are shown in Table \ref{table computer results}, it suggests that the probability of eliminating an edge is asymptotically constant and a quadratic number of edges is eliminated in practice.
\end{rem}

\begin{table}[h!]
  \centering
    \begin{tabular}{l|l}
    \# vertices & edge elimination rate \\
    \hline
    10 & 21.60\%\\
    100 & 16.13\%\\
    1000 & 15.92\%\\
    10000 & 15.89\%\\
    100000 & 15.45\%
    \end{tabular}
    \caption{Probability that an edge is eliminated. For each number of vertices the elimination condition has been checked for $40000$ edges}
       \label{table computer results}
\end{table}

\section{Probabilistic Analysis of Hougardy and Schroeder}
We want to investigate the expected number of edges that remain after the edge elimination of Hougardy and Schroeder. By symmetry it is enough to bound the probability that a fixed edge $pq$ will be deleted.

\subsection{A Modified Criterion}
In this section we modify the criterion of \cite{Hougardy} for detecting useless edges. The new criterion detects fewer useless edges than the original, but it makes the probabilistic analysis easier. This is due to the fact that the new criterion can be separated into several events of similar nature such that the edge is useless if one of these events happens, and for each of these events we can find a bound on its probability. We first state the new edge elimination criterion in Subsection \ref{subsub modified criterion}. The next step is to show that the new criterion has the property that all edges, that will be deleted by it, will be deleted by the original criterion as well. We will show this in  Subsection~\ref{subsub edge elimination}. Recall that in order to apply the original criterion four conditions need to be fulfilled. We begin in Subsection~\ref{subsub general properties} by proving some general geometrical properties we will need later. In Subsection \ref{subsub first condition}, \ref{subsub second condition}, \ref{subsub third condition} and \ref{subsub fourth condition} we show auxiliary lemmas to check each of the four edge elimination properties. Finally, in Subsection~\ref{subsub canonical test region}, we discribe how to repeatly apply the new criterion to increase the probability that $pq$ is eliminated.

\subsubsection{A Modified Criterion} \label{subsub modified criterion}
For every $n$ we consider a fixed $\delta$ defined by $\delta=\delta_n\coloneqq \frac{1}{\sqrt{n}}$. For this choice the area of test regions that will be constructed later is $\frac{1}{n}$, hence the probability that each of them does not contain a vertex is asymptotically constant.
We abbreviate $R_p^{qr} \coloneqq R_p^{qr}(\delta)$ and $R_q^{pr} \coloneqq R_q^{pr}(\delta)$.

For a point $x$ we denote by $x'$ the orthogonal projection of $x$ onto $pq$ in this paper.

\begin{definition}
The \emph{midsegment} of the edge $pq$ is the segment whose endpoints lie on $pq$, which has the same center as $pq$ and whose length is $\frac{1}{2}$ times the length of $pq$. The \emph{neighborhood} of the edge $pq$ is defined as 
\begin{align*}
\{x\in \R^2\mid \dist(x,pq)\leq 3\delta,\ x' \text{ lies on the  midsegment of $pq$}\}.
\end{align*} 
The edge $pq$ divides the neighborhood into two components, we call each of them a \emph{subneighborhood} of $pq$ (Figure \ref{figure region}). 
\end{definition}

\begin{figure}[ht]
\centering
 \definecolor{zzttqq}{rgb}{0.6,0.2,0}
\begin{tikzpicture}[line cap=round,line join=round,>=triangle 45,x=1.0cm,y=1.0cm]
\fill[color=zzttqq,fill=zzttqq,fill opacity=0.1] (5.51,-0.1) -- (9.57,-0.11) -- (9.57,-2.14) -- (5.51,-2.14) -- cycle;
\draw (3.48,-1.12)-- (11.6,-1.12);
\draw (5.51,-2.14)-- (9.57,-2.14);
\draw (9.57,-0.11)-- (5.51,-0.1);
\draw (5.51,-0.1)-- (5.51,-2.14);
\draw (9.57,-0.11)-- (9.57,-2.14);
\draw (4.81,-0.34) node[anchor=north west] {$3\delta$};
\draw (4.83,-1.41) node[anchor=north west] {$3\delta$};
\draw (6.63,0.75) node[anchor=north west] {$\frac{1}{2}\dist(p,q)$};
\begin{scriptsize}
\fill [color=black] (3.48,-1.12) circle (1.5pt);
\draw[color=black] (3.63,-0.89) node {$p$};
\fill [color=black] (11.6,-1.12) circle (1.5pt);
\draw[color=black] (11.75,-0.89) node {$q$};
\end{scriptsize}
\end{tikzpicture}
  \caption{Neighborhood of $pq$}
  \label{figure region}
\end{figure}

\begin{definition}
  Define the distance functions $\dist_x(a,b)=\dist(a',b')$ and $\dist_y(a,b)=\lvert \dist(a,pq)- \dist(b,pq)\rvert$.
\end{definition}
  
We can interpret $\dist_x$ and $\dist_y$ as the distance in $x$ and $y$-direction if we rotate such that $pq$ is parallel to the $x$-axis and $a,b$ lie on the same side of $pq$.

\begin{definition}
Given a point $c$ with $\dist(c,pq)=\delta$, consider the disk $D$ with radius $2\delta$ around $c$. If the orthogonal projection of $D$ onto $pq$ lies in the midsegment of $pq$, we define the \emph{test region} $T$ with center $c$ as $T:=D\cap[0,1]^2$. Consider the two connected components of $\{x\in \R^2 \mid \dist(x,c)\leq \delta, \dist_y(x,c)\geq  \frac{1}{2}\delta \}$, the connected component with smaller distance to $pq$ is called the \emph{lower test subregion $T^l$} and the other the \emph{upper test subregion $T^u$}. (Figure \ref{figure test region}).
\end{definition}

The test region lies in the neighborhood of $pq$, since by definition the orthogonal projection of $T$ onto $pq$ lies in the midsegment and for every point $z\in T$ we have $\dist(z,pq)\leq \dist(z,c)+\dist(c,pq) \leq 3\delta$ which is the height of the neighborhood. The test subregions are constructed to ensure that any two vertices $r\in T^u$ and $s\in T^l$ have at least the distance $\delta$ but at most the distance $2\delta$ to each other.

\begin{figure}[ht]
\centering
  \definecolor{zzttqq}{rgb}{0.6,0.2,0}
\begin{tikzpicture}[line cap=round,line join=round,>=triangle 45,x=1.0cm,y=1.0cm]
\draw (4.,-4.)-- (18.,-4.);
\draw [fill=zzttqq, fill opacity=0.1,shift={(12.163913483825889,-2.9029201124954596)}] plot[domain=0.5235987755982929:2.6179938779915,variable=\t]({1.*1.0970798875045185*cos(\t r)+0.*1.0970798875045185*sin(\t r)},{0.*1.0970798875045185*cos(\t r)+1.*1.0970798875045185*sin(\t r)});
\draw (11.213814431265998,-2.354380168743206)-- (13.11401253638578,-2.354380168743206);
\draw [fill=zzttqq, fill opacity=0.1, shift={(12.163913483825889,-2.9029201124954596)}] plot[domain=3.665191429188086:5.759586531581293,variable=\t]({1.*1.0970798875045187*cos(\t r)+0.*1.0970798875045187*sin(\t r)},{0.*1.0970798875045187*cos(\t r)+1.*1.0970798875045187*sin(\t r)});
\draw (11.213814431265998,-3.4514600562477136)-- (13.11401253638578,-3.4514600562477136);
\draw(12.163913483825889,-2.9029201124954596) circle (2.194159775009081cm);
\draw[dashed] (12.163913483825889,-2.9029201124954596) -- (14.3580732588 ,-2.9029201124954596);
\draw[dashed] (12.163913483825889,-2.9029201124954596) -- (12.163913483825889, -1.8058402250084041);
\draw[dashed] (11.613814431265998,-2.354380168743206) -- (11.613814431265998,-3.4514600562477136);
\begin{scriptsize}
\draw [fill=black] (4.,-4.) circle (1.5pt);
\draw[color=black] (4.0924016743771965,-3.756902590628255) node {$p$};
\draw [fill=black] (18.,-4.) circle (1.5pt);
\draw[color=black] (18.10046909794259,-3.756902590628255) node {$q$};
\draw [fill=black] (7.5,-4.) circle (1.5pt);
\draw[color=black] (7.60046909794259,-3.756902590628255) node {$u$};
\draw [fill=black] (14.5,-4.) circle (1.5pt);
\draw[color=black] (14.60046909794259,-3.756902590628255) node {$v$};
\draw [fill=black] (12.163913483825889,-2.9029201124954596) circle (1.5pt);
\draw[color=black] (12.260330860389269,-3.1008505093727232) node {$c$};
\draw[color=black] (13.2609933713,-2.7029201124954596) node {$2\delta$};
\draw[color=black] (11.413814431265998,-2.9029201124954598) node {$\delta$};
\draw[color=black] (12.363913483825889,-2.628650140623695725) node {$\delta$};
\end{scriptsize}
\end{tikzpicture}
  \caption{Test region and test subregions. The orthogonal projection of the test region onto $pq$ has to lie in the midsegment $uv$ of the edge $pq$.}
  \label{figure test region}
\end{figure}

Now, we can use the test regions to formulate the new criterion that certifies that an edge can be eliminated.

\begin{definition}
We call a test region \emph{certifying} if there is a vertex $r$ in the upper and a vertex $s$ in the lower test subregion and exactly these two vertices in the test region. 
\end{definition}

\begin{definition}
  The edge $pq$ is called $\gamma$-\emph{long} if $\dist(p,q)\geq \gamma\delta=\gamma \frac{1}{\sqrt{n}}$.
\end{definition}
  
Note that if an edge $pq$ is $\gamma$-long then it is also $\gamma_2$-long for all $\gamma_2 \leq \gamma$.

\begin{theorem} \label{strong crit}
  There is a $\gamma$ such that a $\gamma$-long edge $pq$ can be eliminated if there is a certifying test region.
\end{theorem}

We will prove this theorem in Subsection~\ref{subsub edge elimination}.

\subsubsection{General Geometric Properties} \label{subsub general properties}
\begin{lemma} \label{pytha}
For every $\epsilon>0$ there is a $\gamma$ such that if $r$ lies in the neighborhood of a $\gamma$-long edge $pq$, then $\dist(p,r)-\dist(p,r')<\epsilon\delta$ and $\dist(q,r)-\dist(q,r')<\epsilon\delta$.
\end{lemma}

\begin{proof}
We know that $pr$ is the hypotenuse of the right-angled triangle $pr'r$. As $r$ lies in the neighborhood of $pq$, we have $\dist(p,r)\geq \dist(p,r')\geq\frac{\dist(p,q)}{4}\geq \frac{\gamma\delta}{4}$. Thus, by Pythagoras' theorem and $\dist(r,r')\leq 3\delta$:
\begin{align*}
\dist(p,r)-\dist(p,r')&= \sqrt{\dist(p,r')^2+\dist(r,r')^2}-\dist(p,r')\\
 &\leq \sqrt{\dist(p,r')^2+(3\delta)^2}-\dist(p,r')\\
&=\frac{\left(\sqrt{\dist(p,r')^2+(3\delta)^2}\right)^2-\dist(p,r')^2}{\sqrt{\dist(p,r')^2+(3\delta)^2}+\dist(p,r')}\\
&= \frac{(3\delta)^2}{\sqrt{\dist(p,r')^2+(3\delta)^2}+\dist(p,r')}\leq \frac{9}{\sqrt{(\frac{\gamma}{4})^2+9}+\frac{\gamma}{4}}\delta.
\end{align*}
Since $\lim_{\gamma \to \infty} \frac{9}{\sqrt{(\frac{\gamma}{4})^2+9}+\frac{\gamma}{4}}=0$, we can choose $\gamma$ big enough such that we have for all $r$ in the neighborhood of $pq$ $\dist(p,r)-\dist(p,r')<\epsilon\delta$. By an analogous calculation the inequality is also satisfied for $q$. This proves the claim.
\end{proof}

\begin{lemma} \label{beta region}
For all $\beta>0$ there is a $\gamma$ such that if $r$ lies in the neighborhood of a $\gamma$-long edge $pq$, then $\dist(p,r)+\dist(r,q)<\dist(p,q)+\beta\delta$.
\end{lemma}

\begin{proof}
By Lemma \ref{pytha} there is a $\gamma$ such that if $pq$ is $\gamma$-long then $\dist(p,r)-\dist(p,r')<\frac{\beta}{2}\delta$ and $\dist(q,r)-\dist(q,r')<\frac{\beta}{2}\delta$. Adding the two inequalities we get $\dist(p,r)+\dist(r,q)<\dist(p,r')+\frac{\beta}{2}\delta+\dist(r',q)+\frac{\beta}{2}\delta=\dist(p,q)+\beta\delta$.
\end{proof}

\begin{lemma} \label{angle of base}
For every $\epsilon>0$ there is a $\gamma$ such that if the vertex $r$ lies in the neighborhood of a $\gamma$-long edge $pq$ then $\angle qpr< \epsilon$ and $\angle rqp< \epsilon$.
\end{lemma}

\begin{proof}
We may w.l.o.g.\ assume $\epsilon<\pi$. Choose $\gamma$ such that $\frac{12}{\gamma}<\tan(\epsilon)$. By the definition of neighborhood $r'$ lies between $p$ and $q$, hence:
\begin{align*}
\tan(\angle qpr)&= \tan(\angle r'pr)= \frac{\dist(r,r')}{\dist(p,r')}\leq \frac{\dist(r,r')}{\frac{\dist(p,q)}{4}}\leq \frac{4\dist(r,r')}{\gamma \delta} \leq \frac{12\delta}{\gamma\delta}=\frac{12}{\gamma} < \tan(\epsilon).
\end{align*}
Since the tangent is monotonically increasing in $[0,\frac{\pi}{2}]$, we get $\angle qpr < \epsilon$. By the same argument we also have $\angle rqp< \epsilon$.
\end{proof}

\subsubsection{First Condition} \label{subsub first condition}
\begin{lemma} \label{certifying isolated}
If a test region $T$ is certifying, then the unique vertices in the test region $r$ and $s$ are $\delta$-isolated.
\end{lemma}
\begin{proof}
Let $D_r$ and $D_s$ be the disks around $r$ and $s$ with radius $\delta$, respectively. For every point $z\in D_r$ we have $\dist(z,c)\leq \dist(z,r)+\dist(r,c)\leq 2\delta$ where $c$ is the center of the test region $T$. Hence $D_r$ lies in $T$ and similarly $D_s$ also lies in $T$. By assumption $T$ and hence $D_r$ and $D_s$ do not contain any other points except $r$ and $s$. Moreover, we have $\dist(r,s)\geq \dist_y(T^u,T^l)= \delta$, hence $r$ and $s$ are $\delta$-isolated.
\end{proof}

\subsubsection{Second Condition} \label{subsub second condition}
\begin{lemma} \label{atmost angle cone}
For every $\epsilon>0$ there is a $\gamma$ such that if $r$ is lying in the neighborhood of a $\gamma$-long edge $pq$, 
then the angles of the 
cones $R_p^{qr}$ and $R_q^{pr}$ are less than $\epsilon$.
\end{lemma}

\begin{proof}
Consider the circle $\Gamma$ with radius $\delta$ and center$r$. Let $\widetilde{q}$ be the second intersection point of the ray $\overrightarrow{qr}$ with $\Gamma$. As $\dist(\widetilde{q},q)=\dist(r,q)+\delta \geq \dist(p,q)-\dist(p,r) + \delta$, the cone $R_p^{qr}$ contains $\widetilde{q}$ and is not emtpy. 
Let $r^p$ be the point on $R_p^{qr} \cap \Gamma$ that maximizes the angle $\angle \widetilde{q}rr^p$ (Figure \ref{figure most}). 
Since the cone $R_p^{qr}$ is symmetric in $\overrightarrow{qr}$, the ray halves the angle of $R_p^{qr}$ and it is enough to show that $\angle \widetilde{q}rr^p< \frac{\epsilon}{2}$. This is equivalent to $\angle qrr^p> \pi-\frac{\epsilon}{2}$. We can choose by Lemma \ref{beta region} a $\gamma$ such that if $pq$ is $\gamma$-long, we have for all $r$ in the neighborhood of $pq$ the inequality $\dist(p,q)-\dist(p,r)> \dist(r,q)-\beta\delta$ for some constant $\beta <1$ we will determine later. Therefore, by the definition of $R^{qr}_p$ and the cosine law:
\begin{align*}
(\dist(r,q)+(1-\beta) \delta)^2 & < (\dist(p,q)-\dist(p,r)+\delta)^2 \leq \dist(r^p,q)^2 \\
&=\dist(r^p,r)^2+\dist(r,q)^2-2\dist(r^p,r)\dist(r,q)\cos(\angle qrr^p)\\
&= \delta^2+\dist(r,q)^2-2\delta\dist(r,q)\cos(\angle qrr^p).
\end{align*}
Hence, we have:
\begin{align*}
\cos(\angle qrr^p)&< \frac{\delta^2+\dist(r,q)^2-(\dist(r,q)+(1-\beta) \delta)^2}{2\delta\dist(r,q)} \\
&= \frac{\delta^2-2(1-\beta)\delta\dist(r,q)-(1-\beta)^2\delta^2}{2\delta\dist(r,q)} \\
&= (\beta-1)+\frac{(1-(1-\beta)^2)\delta}{2\dist(r,q)}\leq (\beta-1)+\frac{(1-(1-\beta)^2)\delta}{2\frac{\dist(p,q)}{4}} \\
&\leq (\beta-1)+\frac{2(1-(1-\beta)^2)\delta}{\gamma\delta}=(\beta-1)+\frac{2(1-(1-\beta)^2)}{\gamma}.
\end{align*}
Choose $\gamma$ big enough such that by Lemma \ref{beta region} $\beta$ is small enough to satisfy $\cos(\angle qrr^p)< (\beta-1)+\frac{2(1-(1-\beta)^2)}{\gamma}\leq \cos(\pi-\frac{\epsilon}{2})$. Since the cosine is monotonically decreasing in $[0,\pi]$, we have $\angle qrr^p> \pi-\frac{\epsilon}{2}$. Similarly, we have the same result for $R_q^{pr}$.
\end{proof}

\begin{figure}
\centering
 \begin{tikzpicture}[line cap=round,line join=round,>=triangle 45,x=1.0cm,y=1.0cm]
\clip(-7.737225934130034,-5.538755080983481) rectangle (3.6933929431733676,-0.03290724479542577);
\draw (-7.,-4.)-- (3.,-4.);
\draw(-3.0931629766497006,-2.725984990756971) circle (1.6220410440773358cm);
%\draw(3.,-4.) circle (7.512723610391159cm);
%\draw (-7.,-4.)-- (-4.056068883889179,-1.4206768422821325);
\draw [domain=-7.737225934130034:3.0] plot(\x,{(--20.550606878869715--1.274015009243029*\x)/-6.093162976649701});
\draw [domain=-7.737225934130034:-3.0931629766497006] plot(\x,{(--6.662397888627583--1.3053081484748386*\x)/-0.9629059072394788});
\draw [domain=-7.737225934130034:-3.0931629766497006] plot(\x,{(--1.3246800785274004-0.8101569702090314*\x)/-1.4052269682486238});
%\draw [domain=-7.737225934130034:3.6933929431733676] plot(\x,{(-27.25984990756971-0.*\x)/10.});
\begin{scriptsize}
\draw [fill=black] (-7.,-4.) circle (1.5pt);
\draw[color=black] (-6.959225696407813,-3.7673083858621075) node {$p$};
\draw [fill=black] (3.,-4.) circle (1.5pt);
\draw[color=black] (3.0829619874220864,-3.779277620288603) node {$q$};
\draw [fill=black] (-3.0931629766497006,-2.725984990756971) circle (1.5pt);
\draw[color=black] (-3.0093783356642296,-2.510538771080051) node {$r$};
\draw[color=black] (-3.0452860389437166,-0.8946921235031214) node {$\Gamma$};
\draw [fill=black] (-4.056068883889179,-1.4206768422821325) circle (1.5pt);
\draw[color=black] (-4.2900864192992705,-1.3537691562979947) node {$r^p$};
%\draw [fill=black] (-4.498389944898324,-3.5361419609660025) circle (1.5pt);
\draw [fill=black] (-4.680869339956193,-2.394012623552156) circle (1.5pt);
\draw[color=black] (-4.529471107829185,-2.2113079104176565) node {$\widetilde{q}$};
%\draw [fill=black] (-4.715204020727038,-2.725984990756971) circle (1.5pt);
%\draw[color=black] (-4.637194217667646,-2.558415708786034) node {$u$};
\end{scriptsize}
\end{tikzpicture}
  \caption{Sketch for Lemma \ref{atmost angle cone} and Lemma \ref{singlepot}}
  \label{figure most}
\end{figure}

\begin{definition}
  For a given \textsc{Euclidean TSP} instance a vertex $r$ is called $\delta$-\emph{isolated} if for all other vertices $z\neq r$ we have $\dist(r,z) \geq \delta$.
\end{definition}

\begin{lemma} \label{strongpotential}
There exists a $\gamma$ such that if a $\delta$-isolated vertex $r$ lies in the neighborhood of a $\gamma$-long edge $pq$, then neither $R_p^{qr}(\delta)$ nor $R_q^{pr}(\delta)$ contains both neighbors of $r$ in the optimal tour.
\end{lemma}

\begin{proof}
After applying the following lemma, it remains to show that the neighbors of $r$ cannot both lie in $R_p^{qr}$ or $R_q^{pr}$.

\begin{lemmanonumber} [Lemma 7 in \cite{Hougardy}] \label{neighbor in cone}
If a vertex $r$ is $\delta$-isolated, then in every optimal TSP tour containing the edge $pq$ both neighbors of $r$ lie in $R_p^{qr}(\delta) \cup R_q^{pr}(\delta)$.
\end{lemmanonumber}
We prove it similar to Lemma 9 in \cite{Hougardy}. We can choose $\gamma$ by Lemma \ref{atmost angle cone} such that the angles of the cones $R_p^{qr}$ and $R_q^{pr}$ are less than $60^\circ$. By Lemma \ref{beta region} there is a $\gamma$ such that if $pq$ is $\gamma$-long, then $\dist(p,r)+\dist(r,q)-\dist(p,q)< \delta$.
Thus, we can use the following lemma:

\begin{lemmanonumber}[Lemma 10 in \cite{Hougardy}] \label{angle of neighbors}
  Let $x$ and $y$ be the neighbors of a $\delta$-isolated vertex $r$ in an optimal TSP tour containing the edge $pq$. If $\dist(p,r)+\dist(r,q)-\dist(p,q)\leq 2\delta$, then we have:
  \begin{align*}
  \angle xry \geq \arccos \left(1-\frac{(2\delta+\dist(p,q)-\dist(p,r)-\dist(r,q))^2}{2\delta^2} \right).
  \end{align*}
\end{lemmanonumber}

Therefore,
\begin{align*}
\angle xry &\geq \arccos \left(1-\frac{(2\delta+\dist(p,q)-\dist(p,r)-\dist(r,q))^2}{2\delta^2} \right)\\
&>\arccos \left(1-\frac{\delta^2}{2\delta^2} \right)=\arccos(\frac{1}{2})=60^\circ,
\end{align*} contradiction.
\end{proof}

\subsubsection{Third Condition} \label{subsub third condition}

\begin{lemma}\label{r not in cone}
There is a $\gamma$ such that for every $\gamma$-long edge $pq$, test region $T$ and vertices $r, s$ with $r\in T^u$ and $s\in T^l$ we have $r\not\in R_p^{qs}\cup R_q^{ps}$.
\end{lemma}

\begin{proof}
Let $p_s$ be the intersection point of $ps$ with the circle around $s$ and radius $\delta$. 
Then, we have:
\begin{align*}
\dist(q,p_s) &= \dist(q,p_s)+\dist(p_s,p)-\dist(p_s,p) \geq \dist(p,q)-\dist(p_s,p)\\
&=\dist(p,q)+\delta-\dist(s,p).
\end{align*}
Therefore, we have by definition $p\in R_p^{qs}$. We have by the definition of test region and monotonicity of the tangent function $\angle{rsp} \geq \arctan\left(\frac{\dist_y(r,s)}{\dist_x(r,s)}\right) \geq \arctan(\frac{\delta}{2\delta})= \arctan(\frac{1}{2})$. By Lemma \ref{atmost angle cone} we can choose $\gamma$ big enough such that the angle of the cone $R_p^{qs}$ is less than $\arctan(\frac{1}{2})$. Thus, we have $r\not \in R_p^{qs}$ and similarly $r\not\in R_q^{ps}$.
\end{proof}

\subsubsection{Fourth Condition}\label{subsub fourth condition}
\begin{lemma} \label{singlepot}
Let $\Gamma$ be the circle around $r$ with radius $\delta$ and $r^p$ is the point on $R^{qr}_p\cap \Gamma$ with the largest distance to $p$. Let $r^q$ be defined similarly. There is a $\gamma$ such that if a $\delta$-isolated vertex $r$ lies in a test region with center $c$ of a $\gamma$-long edge $pq$, then: 
  \begin{align*}
  \dist(p,r^p) &<\dist_x(p,c)\\
  \dist(r^q,q) &<\dist_x(c,q).
  \end{align*}
\end{lemma}

\begin{proof}
Let $\widetilde{q}$ be the second intersection point of the ray $\overrightarrow{qr}$ and $\Gamma$. (Figure \ref{figure most}).  Choose $\epsilon_1, \epsilon_2>0$ such that $\frac{\sqrt{3}}{2}+\epsilon_1+\epsilon_2<1$. Since $\widetilde{q}$ lies on the circle with radius $\delta$ around $r$, it lies in the neighborhood of $pq$. So we can choose by Lemma~\ref{pytha} a $\gamma_1$ such that if $pq$ is $\gamma_1$-long, then 
\begin{align}
\dist(p,r^p)-\dist_x(p,r^p)<\epsilon_1\delta. \label{eq hyp distance p rp}
\end{align}
By Lemma \ref{atmost angle cone}, there is a $\gamma_2$ such that if $pq$ is $\gamma_2$-long the angle of the cone $R_p^{qr}$ is small enough to ensure
\begin{align}
\dist_x(r^p,\widetilde{q})\leq \epsilon_2 \delta. \label{eq div rp widetilde q}
\end{align}
By Lemma \ref{angle of base}, we can choose $\gamma_3$ such that if $pq$ is $\gamma_3$-long, then 
\begin{align}
\dist_x(\widetilde{q},r) = \cos(\angle \widetilde{q}qp)\dist(\widetilde{q},r)=\cos(\angle rqp)\delta>\left(\frac{\sqrt{3}}{2}+\epsilon_1+\epsilon_2 \right) \delta. \label{dist x widetilde q r}
\end{align}
As by definition of test regions $\dist_x(r,c)^2+\dist_y(r,c)^2\leq \delta^2$ and $\dist_y(r,c) \geq \frac{\delta}{2}$, we have 
\begin{align}
\dist_x(r,c) \leq \frac{\sqrt{3}}{2}\delta. \label{eq dist x rc}
\end{align}
Therefore, for $\gamma=\max\{\gamma_1,\gamma_2, \gamma_3\}$ we get altogether: 
\begin{align*}
\dist(p,r^p)&\overset{\ref{eq hyp distance p rp}}{<}\dist_x(p,r^p)+\epsilon_1\delta\leq \dist_x(p,\widetilde{q})+\dist_x(\widetilde{q},r^p)+\epsilon_1\delta
\overset{\ref{eq div rp widetilde q}}{\leq} \dist_x(p,\widetilde{q})+\epsilon_1\delta+\epsilon_2\delta \\
&\overset{\ref{dist x widetilde q r}}{<} \dist_x(p,\widetilde{q})+\epsilon_1\delta+\epsilon_2\delta + \dist_x(\widetilde{q},r) - \left(\frac{\sqrt{3}}{2}+\epsilon_1+\epsilon_2\right) \delta\\
&=\dist_x(p,r) -\dist_x(r,c) + \dist_x(r,c) - \frac{\sqrt{3}}{2} \delta \overset{\ref{eq dist x rc}}{\leq} \dist_x(p,c).
\end{align*}
Similarly, we can prove the second statement.
\end{proof}

\subsubsection{Proof of Theorem \ref{strong crit}} \label{subsub edge elimination}

\begin{proof}[Proof of Theorem \ref{strong crit}]
We want to apply Theorem~\ref{main elimination} to eliminate the edge $pq$. Recall that we have to check that four conditions are fulfilled.
By Lemma~\ref{certifying isolated} and Lemma~\ref{strongpotential} there is a $\gamma_1$ such that if $pq$ is $\gamma_1$-long the first two conditions are satisfied.

In the proof of Theorem \ref{main elimination} we need $r\not\in R_p^{qs} \cup R_q^{ps}$ and $s\not\in R_p^{qr}\cup R_q^{pr}$ to exclude that $rs$ is part of the tour. Note that this already follows from $r\not\in R_p^{qs} \cup R_q^{ps}$, since we know that every neighbor of $s$ lies in $R_p^{qs} \cup R_q^{ps}$. By Lemma \ref{r not in cone} there is a $\gamma_2$ such that if $pq$ is $\gamma_2$-long, then we have $r\not\in R_p^{qs} \cup R_q^{ps}$, which is the third condition.

Let $c$ be the center of the test region. 
By Lemma \ref{singlepot} there is a $\gamma_3$ such that for a $\gamma_3$-long edge $pq$ we have:
\begin{align*}
&\dist(p,r^p)+\dist(s^q,q)+\dist(r,s) <\dist_x(p,c)+\dist_x(c,q)+\dist(r,s)=\dist(p,q)+\dist(r,s) \\
\leq &\dist(p,q)+\dist(r,c) + \dist(c,s) \leq \dist(p,q)+2\delta = \dist(p,q)+ \dist(r^p,r) + \dist(s,s^q).
\end{align*}
Similarly, we can ensure $\dist(p,s^p)+\dist(r^q,q)+\dist(r,s) < \dist(p,q)+ \dist(s^p,s)+\dist(r,r^q)$, which is the fourth conditon. Hence, $\gamma=\max\{\gamma_1,\gamma_2,\gamma_3\}$ satisfies the condition.
\end{proof}

\subsubsection{Canonical Test Regions} \label{subsub canonical test region}
Now we want to define for each edge a certain number of test regions. The number of test regions is dependent on the length of the edge. An edge can be eliminated if any of these is certifying.

\begin{lemma} \label{insidebox}
There exists a $\gamma$ such that given a $\gamma$-long edge $pq$, there is a subneighborhood of $pq$ that lies completely inside the unit square.
\end{lemma}

\begin{proof}
We can w.l.o.g.\ assume that the angle formed by the line $pq$ with the $x$-axis is $0^\circ \leq \alpha \leq 45^\circ$, otherwise we can rotate and mirror the unit square. In this case, there is a constant $c>0$ such that the $x$-coordinates of the midsegment lie between $c$ and $1-c$. Given that $\sqrt{2} \geq \dist(p,q) \geq \gamma\delta$, we can choose a $\gamma$ such that if $pq$ is $\gamma$-long, the $x$-coordinates of both subneighborhoods lie between $0<c-3\delta$ and $1-c+3\delta<1$. Since the range of $y$-coordinates of the midsegment is at most $\frac{1}{2}$, we can choose $\gamma$ large enough to ensure that the range of $y$-coordinates of the neighborhood is at most $6\delta+\frac{1}{2}<1$. Then, it is not possible for both the top and bottom borders of the unitsquare to intersect the neighborhood. Since the borders of the unit square do not intersect $pq$, which divides the two subneighborhoods, one of the subneighborhoods lies completely within the unit square. 
\end{proof}

\begin{definition}
We denote the smallest $\gamma$ satisfying the conditions of Theorem \ref{strong crit} and Lemma \ref{insidebox} by $\eta$. If an edge is $\eta$-long we call it \emph{long}. 
\end{definition}

The practical algorithm in \cite{Hougardy} checks the edge elimination condition for 10 pairs of heuristically chosen $r$ and $s$. We consider the general setting of checking at most $f(n)$ pairs for some function $f:\N \to \N$. We now assume that such a function is given. In the end we will show that any $f\in \omega(\log(n)) \cap o(n)$ is sufficient to conclude that the expected value of remaining edges is linear.

\begin{definition}
For every long edge $pq$ we divide the midsegment $uv$ equally into $\min\{f(n), \lfloor \frac{\dist(u,v)}{4\delta} \rfloor\}$ parts. By Lemma \ref{insidebox} there is a subneighborhood that lies inside of the unit square. Place inside of this subneighborhood test regions for each of the subdivisions of $uv$ such that the orthogonal projections of the test regions' centers coincide with the centers of the subdivisions (Figure \ref{figure test area}). We call the constructed test regions $\mathcal{T}=\{T_1,\dots, T_{\lvert \mathcal{T} \rvert}\}$ the \emph{canonical test regions} of $pq$. The $i$-th canonical test region $T_i$ contains $T_i^u$ and $T_i^l$, the upper and lower test subregion.
\end{definition}

The subneighborhood has height $3\delta$ and the upper and lower test subregions have at most distance $2\delta$ to $pq$, therefore the constructed upper and lower test subregions lie completely in the subneighborhood and hence in the unit square. Remember that the test region was defined as the intersection of a disk with the unit square and note that the test region of a canonical test region is not necessarily a disk.

Since $pq$ has at least length $\eta\delta$, the midsegment has at least length $\frac{\eta\delta}{2}$. Therefore, we constructed at least $\frac{\eta\delta}{2\cdot4\delta}=\frac{\eta}{8}$ canonical test regions for every long edge.

\begin{figure}[ht]
\centering
 \begin{tikzpicture}[line cap=round,line join=round,>=triangle 45,x=1.0cm,y=1.0cm]
\draw (0.,-3.)-- (14.,-3.);
\draw [shift={(4.6797346193001115,-2.468999167977505)}] plot[domain=0.523598775598309:2.617993877991484,variable=\t]({1.*0.531000832022491*cos(\t r)+0.*0.531000832022491*sin(\t r)},{0.*0.531000832022491*cos(\t r)+1.*0.531000832022491*sin(\t r)});
\draw (4.219874409337963,-2.2034987519662548)-- (5.1395948292622595,-2.2034987519662548);
\draw [shift={(4.6797346193001115,-2.468999167977505)}] plot[domain=3.665191429188102:5.759586531581277,variable=\t]({1.*0.531000832022491*cos(\t r)+0.*0.531000832022491*sin(\t r)},{0.*0.531000832022491*cos(\t r)+1.*0.531000832022491*sin(\t r)});
\draw (4.219874409337963,-2.734499583988755)-- (5.1395948292622595,-2.734499583988755);
\draw(4.6797346193001115,-2.468999167977505) circle (1.062001664044982cm);
\draw [shift={(9.345709248860153,-2.4523349728719332)}] plot[domain=0.523598775598306:2.6179938779914886,variable=\t]({1.*0.5310008320224903*cos(\t r)+0.*0.5310008320224903*sin(\t r)},{0.*0.5310008320224903*cos(\t r)+1.*0.5310008320224903*sin(\t r)});
\draw (8.885849038898003,-2.186834556860685)-- (9.805569458822301,-2.186834556860685);
\draw [shift={(9.345709248860153,-2.4523349728719332)}] plot[domain=3.6651914291880976:5.75958653158128,variable=\t]({1.*0.5310008320224917*cos(\t r)+0.*0.5310008320224917*sin(\t r)},{0.*0.5310008320224917*cos(\t r)+1.*0.5310008320224917*sin(\t r)});
\draw (8.885849038898003,-2.7178353888831817)-- (9.805569458822301,-2.7178353888831817);
\draw(9.345709248860153,-2.4523349728719332) circle (1.0620016640449834cm);
\draw [shift={(7.0127219340801314,-2.4523349728719337)}] plot[domain=0.5235987755982808:2.6179938779915126,variable=\t]({1.*0.5310008320224947*cos(\t r)+0.*0.5310008320224947*sin(\t r)},{0.*0.5310008320224947*cos(\t r)+1.*0.5310008320224947*sin(\t r)});
\draw (6.552861724117973,-2.1868345568606946)-- (7.47258214404229,-2.1868345568606946);
\draw [shift={(7.0127219340801314,-2.4523349728719337)}] plot[domain=3.6651914291880723:5.759586531581307,variable=\t]({1.*0.5310008320224944*cos(\t r)+0.*0.5310008320224944*sin(\t r)},{0.*0.5310008320224944*cos(\t r)+1.*0.5310008320224944*sin(\t r)});
\draw (6.552861724117973,-2.717835388883172)-- (7.47258214404229,-2.717835388883172);
\draw(7.0127219340801314,-2.4523349728719337) circle (1.0620016640449896cm);
\draw (3.5,-3.)-- (3.5,-1.406997503932524);
\draw (3.5,-1.406997503932524)-- (10.5,-1.406997503932524);
\draw (10.5,-1.406997503932524)-- (10.5,-3.);
\begin{scriptsize}
\draw [fill=black] (0.,-3.) circle (1.5pt);
\draw[color=black] (0.11374516037349916,-2.8022978994555076) node {$p$};
\draw [fill=black] (14.,-3.) circle (1.5pt);
\draw[color=black] (14.111669049053624,-2.8022978994555076) node {$q$};
\draw [fill=black] (3.5,-3.) circle (1.5pt);
\draw[color=black] (3.35,-2.8022978994555076) node {$u$};
\draw [fill=black] (10.5,-3.) circle (1.5pt);
\draw[color=black] (10.611669049053624,-2.8022978994555076) node {$v$};
\draw [fill=black] (4.6797346193001115,-2.468999167977505) circle (1.5pt);
\draw[color=black] (4.963025936094828,-2.45685607132933) node {$c_1$};
\draw [color=black] (5.833333333333334,-3.) circle (1.5pt);
\draw [color=black] (8.166666666666668,-3.) circle (1.5pt);
\draw [fill=black] (9.345709248860153,-2.4523349728719332) circle (1.5pt);
\draw[color=black] (9.62900056565487,-2.459021412027361) node {$c_3$};
\draw [fill=black] (7.0127219340801314,-2.4523349728719337) circle (1.5pt);
\draw[color=black] (7.296013250874849,-2.459021412027361) node {$c_2$};
\end{scriptsize}
\end{tikzpicture}
  \caption{The canonical test regions. The open dots are the subdivisions of the midsegment $uv$.}\label{figure test area}
\end{figure}

Every canonical test region has width $4\delta$ and by the construction the distances of their midpoints is at least $4\delta$. Therefore, the constructed canonical test regions intersect in at most one point.

Alltogether, we can formulate the new criterion:

\begin{corollary}
A long edge is useless if there is a certifying canonical test region $T_i$.
\end{corollary}

\subsection{Probabilistic Analysis}
In this subsection we estimate the probability that an edge cannot be deleted by our modified criterion on a random instance and bound the expected value of the number of edges that are not detected as useless.

\begin{definition}
  For a subset $C$ of the unit square we define $\#(C)$ as the number of vertices in $C$.
\end{definition}

\begin{definition}
Let $A_i$ and $B_i$ be events defined as follows:
\begin{itemize}
\item $A_i$ is the event $\#(T^l_i)=0 \lor \#( T^u_i)=0$
\item $B_i$ is the event $\#(T_i^u) \geq 1 \land \#(T_i^l) \geq 1 \land \#(T_i) \geq 3$.
\end{itemize}
\end{definition}

Note that $\Pr\left[A_i\right]$ and $\Pr\left[B_i\right]$ are dependent on the number of vertices $n$ of the random instance and on the position of $T_i$. 
Observe that the two events are disjoint and if $A_i^{c}\land B_i^{c}$ happens, then the test region $T_i$ is certifying and therefore $pq$ is useless. 

Our next aim is to show the almost mutually independence of $A_i \lor B_i$ for $T_i\in \mathcal{T}$.

\begin{lemma} \label{lim converging}
Given $J$, $K\subseteq \{1,\dots, \lvert \mathcal{T} \rvert\}\backslash\{i\}$ with $J\cap K = \emptyset$. Then,
\begin{align*}
\Pr[A_i \mid \land_{j\in J}A_j] & \leq \Pr[A_i] \\
\limsup_{n\to\infty} \Pr[B_i \mid (\land_{j\in J}A_j) \land (\land_{k\in K} B_k)] & \leq \limsup_{n\to\infty} \Pr[B_i].
\end{align*}
\end{lemma}

\begin{proof}
For the first statement, note that for a set $C\in \{T_i^u, T_i^l\}$ we have:
\begin{align*}
&\Pr[x\in C \mid \land_{j\in J} (\#(T_j^u)=0 \lor \#(T_j^l)=0)]\geq \Pr[x\in C]
\end{align*}
As $A_i = \#(T_i^u)=0 \lor \#(T_i^l)=0$, we conclude $\Pr[A_i \mid \land_{j\in J}A_j] \leq \Pr[A_i]$.

To show the second statement, observe that for a set $C\in \{T_i^u, T_i^l,T_i\}$ we have:
\begin{align*}
&\Pr[x\in C \mid (\land_{j\in J}A_j) \land (\land_{k\in K} B_k)]\\
=&\Pr[x\in C \mid (\land_{j\in J} (\#(T_j^u)=0 \lor \#(T_j^l)=0)) \land (\land_{k\in K} (\#(T_k^u)\geq 1 \land \#(T_k^l)\geq 1 \land \#(T_k)\geq 3))] \\
\leq &\Pr[x\in C \mid \land_{j\in J} (\#(T_j^u)=0 \land \#(T_j^l)=0)]\\
\leq &\Pr[x\in C \mid \#(\mathcal{T}\backslash T_i)=0]=\Pr[x\in C \mid x\not\in(\mathcal{T}\backslash T_i)]\\
=&\frac{\Pr[x\in C \land x\not\in(\mathcal{T}\backslash T_i)]}{\Pr[x\not\in(\mathcal{T}\backslash T_i)]} = \frac{\Pr[x\in C]}{\Pr[x\not\in(\mathcal{T}\backslash T_i)]} \leq \frac{\Pr[x\in C]}{1-\phi\lambda(\mathcal{T}\backslash T_i)}
\end{align*}
Note that $\phi\lambda(\mathcal{T}\backslash T_i) = \Theta\left(\frac{f(n)}{n}\right)\to 0$ as $n\to \infty$. Hence, we have $\limsup_{n\to \infty} \Pr[x\in C \mid (\land_{j\in J}A_j) \land (\land_{k\in K} B_k)]\leq \limsup_{n\to \infty} \Pr[x\in C]$. Because of $B_i = \#(T_i^u)\geq 1 \land \#(T_i^l)\geq 1 \land \#(T_i)\geq 3$, we conclude $\limsup_{n\to\infty} \Pr[B_i \mid (\land_{j\in J}A_j) \land (\land_{k\in K} B_k)] \leq \limsup_{n\to\infty} \Pr[B_i]$.
\end{proof}

If $pq$ cannot be deleted, all test regions in $\mathcal{T}$ are not certifying. That means for each test region either $A_i$ or $B_i$ happens. 
As these events are negatively correlated by Lemma \ref{lim converging}, we conclude:
\begin{corollary} \label{formular for single WK}
Given an edge $pq$ we have:
\begin{align*}
\limsup_{n\to\infty}\Pr\left[X_n^{pq}=1 \mid \lvert \mathcal{T} \rvert=k\right] &\leq \limsup_{n\to\infty} \prod_{i=1}^k (\Pr\left[A_i\right]+\Pr\left[B_i\right]).
\end{align*}
\end{corollary}

By construction, we know that $p$ and $q$ do not lie in any $T_i$ and do not influence $A_i$ or $B_i$ for any $i$. Therefore, we can rename the vertices and assume that $p=U_{n-1}$, $q=U_n$, discard $p$ and $q$ for simplicity and just consider the remaining $n-2$ vertices in $[0,1]^2$. The next step to bound the probability that one test region is certifying asymptotically by a constant.

\begin{lemma} \label{AS1}
Independent on the position of $T_i$ we have for all $i$:
\begin{align*}
\limsup_{n\to \infty}\Pr\left[A_i\right]<1.
\end{align*}
\end{lemma}

\begin{proof}
We have
\begin{align*}
\Pr\left[A_i\right]&=\Pr\left[\#(T_i^u)=0\lor \#(T_i^l)=0\right]\\
&=\Pr\left[\#(T_i^u)=0\right]+\Pr\left[\#(T_i^l)=0\right]- \Pr\left[\#(T_i^u)=0\land \#(T_i^l)=0\right]\\
&=\Pr\left[\#(T_i^u)=0\right]+\Pr\left[\#(T_i^l)=0\right]- \Pr\left[\#(T_i^u \cup T_i^l)=0\right]\\
&=\prod_{j=1}^{n-2}{\int_{[0,1]^2\backslash T_i^u}d_j\mathrm{d}\lambda}+\prod_{j=1}^{n-2}{\int_{[0,1]^2\backslash T_i^l}d_j\mathrm{d}\lambda}-\prod_{j=1}^{n-2}{\int_{[0,1]^2\backslash (T_i^u\cup T_i^l)}d_j\mathrm{d}\lambda}.
\end{align*}
Define $y_j \coloneqq \int_{T_i^u}d_j\mathrm{d}\lambda, z_j \coloneqq \int_{T_i^l}d_j\mathrm{d}\lambda$. Using the linearity of the integral and $\int_{[0,1]^2\backslash (T_i^u\cup T_i^l)}d_j\mathrm{d}\lambda+\int_{T_i^u}d_j\mathrm{d}\lambda+\int_{T_i^l}d_j\mathrm{d}\lambda=\int_{[0,1]^2}d_j\mathrm{d}\lambda=1$ we get:
\begin{align*}
\Pr\left[A_i\right]=\prod_{j=1}^{n-2}(1-y_j)+\prod_{j=1}^{n-2}(1-z_j)-\prod_{j=1}^{n-2}(1-y_j-z_j).
\end{align*}
Now note that since $y_j,z_j\geq 0$ we have for all $k$:
\begin{align*}
\frac{\partial}{\partial y_k} \Pr\left[A_i\right]&= -\prod_{j=1,j\neq k}^{n-2}(1-y_j)+ \prod_{j=1,j\neq k}^{n-2}(1-y_j-z_j)\leq 0\\
\frac{\partial}{\partial z_k} \Pr\left[A_i\right]&= -\prod_{j=1,j\neq k}^{n-2}(1-z_j)+ \prod_{j=1,j\neq k}^{n-2}(1-y_j-z_j)\leq 0.
\end{align*}
We conclude that $\Pr\left[A_i\right]$ is monotonically decreasing in $y_k$ and $z_k$. Define $F:=n\lambda(T^u)$, since $\lambda(T^u)=\lambda(T^l)$ scales linearly with $\delta^2=\frac{1}{n}$, $F$ is a constant. Moreover, we have $\lambda(T^u)=\lambda(T^l)=\frac{F}{n}$.
We know that $y_k=\int_{T_i^u}d_k\mathrm{d}\lambda\geq \psi\lambda(T_i^u)=\psi\frac{F}{n}$ and similarly $z_k=\int_{T_i^l}d_j\mathrm{d}\lambda\geq \psi\frac{F}{n}$. Thus

\begin{align*}
\Pr\left[A_i\right]&\leq 2(1- \psi\frac{F}{n})^{n-2}-(1- 2\psi\frac{F}{n})^{n-2}\\
\Rightarrow \limsup_ {n\to \infty}\Pr\left[A_i\right] &\leq \limsup_{n\to \infty}2(1- \psi\frac{F}{n})^{n-2}-(1- 2\psi\frac{F}{n})^{n-2}=2e^{-\psi F}-e^{-2\psi F}\\
&=1-1+2e^{-\psi F}-e^{-2\psi F}=1-(1-e^{-\psi F})^2<1.
\end{align*}
In the last inequality we used that $\psi F>0$.
\end{proof}

\begin{lemma} \label{LeastS1}
Independent on the position of $T_i$ we have for all $i$:
\begin{align*}
\limsup_{n\to \infty} \Pr\left[\#(T_i) \geq 3\mid A_i^c\right]<1.
\end{align*}
\end{lemma}

\begin{proof}
We have
\begin{align*}
\Pr\left[\#(T_i)\geq 3 \mid A_i^c\right]&=\Pr\left[\#(T_i)\geq 3 \mid \#(T_i^u)\geq 1\land \#(T_i^l)\geq 1\right]\\
&=\Pr\left[\#(T_i) \geq 3 \mid \#(T_i)\geq 2\right]\leq \max_{x,y} \Pr\left[\#(T_i)\geq 3 \mid U_x,U_y\in T_i\right]\\
&=\max_{x,y}\Pr\left[\#(T_i) \geq 1 \mid U_x,U_y\not\in T_i\right]\\
&=1-\min_{x,y}\Pr\left[\#(T_i)=0 \mid U_x,U_y\not\in T_i\right]\\
&\leq 1-\min_{x,y}\prod_{j=1,j\neq x,y}^{n-2} \int_{[0,1]^2\backslash T_i}d_j\mathrm{d}\lambda=1-\min_{x,y}\prod_{j=1,j\neq x,y}^{n-2} \left( 1- \int_{T_i}d_j\mathrm{d}\lambda \right)\\
&\leq 1-(1-\phi\lambda(T_i))^{n-4} \leq 1-(1-\phi\frac{4\pi}{n})^{n-4}.
\end{align*}
In the last inequality we bound the area of $T_i$ by the area of the circle with radius $2\delta$. Hence $\limsup_{n\to \infty} \Pr\left[\#(T_i)\geq 3 \mid A_i^c\right] \leq \limsup_{n\to \infty} 1-(1-\phi\frac{4\pi}{n})^{n-4}=1-e^{-4\pi\phi}< 1$.
\end{proof}

\begin{lemma} \label{ep}
There exists a $0<s<1$ such that, independent of the position of the canonical test regions, for all $i$:
\begin{align*}
\limsup_{n\to\infty} \Pr\left[A_i\right]+\Pr\left[B_i\right]=s<1.
\end{align*}
\end{lemma}

\begin{proof} We have:
\begin{align*}
&\limsup_{n \to \infty} \Pr\left[A_i\right]+\Pr\left[B_i\right]\\
= & \limsup_{n \to \infty} \Pr\left[A_i\right]+\Pr\left[A_i^c\right] \Pr\left[\#(T_i) \geq 3 \mid A_i^c\right]\\
=&\limsup_{n\to \infty} 1- (1-\Pr\left[A_i\right]) (1-\Pr\left[\#(T_i)\geq 3 \mid A_i^c\right]).
\end{align*}
We know by Lemma \ref{AS1} and \ref{LeastS1} that there are positive constants $c_1, c_2$ such that $\limsup_{n\to \infty} (1-\Pr\left[A_i\right])>c_1$ and $\limsup_{n\to \infty} \Pr\left[ \#( T_i ) \geq 3 \mid A_i^c\right]<1-c_2$. Hence, we have $\limsup_{n\to \infty} \Pr\left[A_i\right]+\Pr\left[B_i \mid \#( \mathcal{T} \backslash T_i ) =0\right]<1-c_1c_2<1$.
\end{proof}

With this estimate we can now investigate the asymptotic behavior of the number of remaining edges after the edge elimination.

\begin{theorem} \label{Hougardy result}
We have $\E\left[X_n\right]\in \Theta(n)$.
\end{theorem}

\begin{proof}
We have by Corollary \ref{formular for single WK} and Lemma \ref{ep} for large $n$:
\begin{align*}
\Pr\left[X_n^{pq}=1\right]&=\sum_{l=0}^{\infty}\Pr\left[X_n^{pq}=1 \mid \lvert\mathcal{T}\rvert=l\right]\Pr\left[\lvert\mathcal{T}\rvert=l\right]\\
&\leq s^0\Pr\left[\lvert\mathcal{T}\rvert=0\right]+ \sum_{l=\lfloor \eta/8 \rfloor}^{f(n)-1}s^l\Pr\left[\lvert\mathcal{T}\rvert=l\right]+ s^{f(n)}\Pr\left[\lvert\mathcal{T}\rvert=f(n)\right].
\end{align*}
By the construction of the canonical test regions $\Pr\left[\lvert\mathcal{T}\rvert=0\right]$ is the probability that $pq$ has length less than $\eta\delta$ and $\Pr\left[\lvert\mathcal{T}\rvert=l\right]$ for $\frac{\eta}{8} \leq l< f(n)$ is the probability that the random edge $pq$ has length between $8\delta l$ and $8\delta(l+1)$. The first event is equivalent to that $q$ lies in a circle with radius $\eta\delta=\frac{\eta}{\sqrt{n}}$ around $p$. This has probability hat most $\pi\frac{\eta^2}{n}\phi$. The second event is equivalent to $q$ is not lying in the circle with radius $8l\delta=\frac{8l}{\sqrt{n}}$ around $p$ but in the circle with radius $8(l+1)\delta=\frac{8(l+1)}{\sqrt{n}}$ around $p$. This probability is less or equal to $64\pi\frac{2l+1}{n}\phi$, since this is an upper bound for the area of the annulus intersected with $[0,1]^2$ times the upper bound for the density function. 
Thus
\begin{align*}
\E\left[X_n\right]&=\E\left[\sum_{e\in E(K_n)}X^{e}_n\right]=\sum_{e\in E(K_n)}\E\left[X^{e}_n\right]=\sum_{e\in E(K_n)}\Pr\left[X^{e}_n=1\right]\leq \binom{n}{2}\Pr\left[X^{pq}_n=1\right]\\
&\leq \binom{n}{2}\left(\pi\frac{\eta^2}{n}\phi+\sum_{l=0}^{f(n)-1}s^l64\pi\frac{(2l+1)}{n}\phi+s^{f(n)} \right).
\end{align*}

Choose some $f$ satisfying $f(n)\in\omega(\log n)$ and $f(n)\in o(n)$. 
\begin{align*}
&\limsup_{n\to \infty}\frac{\E\left[X_n\right]}{n} =\lim_{n\to \infty} \frac{n-1}{2n} \lim_{n\to \infty} n\left(\pi\frac{\eta^2}{n}\phi+ \sum_{l=0}^{f(n)-1}s^l64\pi\frac{(2l+1)}{n}\phi+s^{f(n)} \right)\\
=&\frac{1}{2}\left(\pi\eta^2\phi+ \lim_{n\to \infty} \sum_{l=0}^{f(n)-1}s^l64\pi(2l+1)\phi+\lim_{n\to \infty}s^{f(n)}n\right)\\
=&\frac{1}{2}\left(\pi\eta^2\phi+ \sum_{l= 0}^{\infty}s^l64\pi(2l+1)\phi\right)\\
<&\infty.
\end{align*}
The last expression is a constant because $\limsup_{l\to \infty} \sqrt[l]{s^l64\pi(2l+1)\phi}=s<1$, hence by the root criterion the series $\sum_{l=0}^{\infty}s^l64\pi(2l+1)\phi$ converges.
Thus, $\E\left[X_n\right]\in O(n)$. On the other hand, we have clearly $\E\left[X_n\right]\in \Omega(n)$, since an optimal \textsc{TSP} tour exists for every instance and consists of $n$ edges which are thus not useless. Therefore, we conclude $\E\left[X_n\right]\in \Theta(n)$.
\end{proof}

\begin{rem}
For a uniformly random distribution of the vertices we get with this proof an upper bound of roughly $3.73\cdot 10^{14} n$ on the expected number of remaining edges. The first summand $\frac{1}{2}\pi\frac{\eta^2}{n}\phi$ from the constant factor is dominated by the second summand $\frac{1}{2}\left(\sum_{l=0}^{f(n)-1}s^l64\pi\frac{(2l+1)}{n}\phi\right)$ as $\eta=100$ suffice for the geometric properties. A straightforward calculation shows that $\lambda(T_i^u)=\lambda(T_i^l)= \frac{\frac{2\pi}{3}-\frac{\sqrt{3}}{2}}{2n}$. Using this we can bound the probability that one test region is not certifying by $s\leq 0.9999992655$. Note that the purpose of this proof is mainly to show the asymptotic behavior and not to optimize the constant.
\end{rem}

\subsubsection*{Acknowledgements}
I would like to express my gratitude to my supervisor Stefan Hougardy for introducing me to this topic, giving advice and suggesting corrections. Furthermore, I want to thank Fabian Henneke, Heiko Röglin and anonymous reviewers for reading this paper and making helpful remarks.

\bibliographystyle{plain}
\bibliography{ProbabilisticAnalysisOfEdgeElimination.bbl}
%\bibliography{myBib}
\end{document}